%% file: main.tex
\documentclass[11pt]{article}

\input{myheader}

\newcommand{\PL}{\hbox{\rm\small PL}}

\usepackage[most]{tcolorbox}
\newtcolorbox{updatebox}[1][Update:]{
  breakable,
  colback=yellow!10,
  colbacktitle=yellow!10,
  colframe=black!60,
  coltext=black,
  boxrule=0.5pt,
  arc=0pt,
  left=6pt,
  right=6pt,
  top=5pt,
  bottom=5pt,
  before skip=8pt,
  after skip=8pt,
  fontupper=\normalfont,
  title={#1},
  coltitle=black!80,
  fonttitle=\footnotesize\bfseries,
  titlerule=0pt,
  toptitle=1.5pt,
  bottomtitle=1.5pt,
  lefttitle=6pt,
  righttitle=6pt,
}

\date{}

\begin{document}

%\bigskip
\begin{center}
{\large SIGACT News Complexity Theory Column, March 2024\\ \medskip \bf \Large
 Structure in Communication Complexity and Constant-Cost Complexity Classes}\\
%
%\vskip -0.1in
\bigskip 

{\Large  
Hamed Hatami}\footnote{McGill University, Montreal, QC, Canada \@. {\tt hatami@cs.mcgill.ca}. 
Supported by an NSERC grant.}\ \ \ \ 
{\Large  
Pooya Hatami}\footnote{The Ohio State University, Columbus, OH, USA\@. {\tt pooyahat@gmail.com}. Supported by NSF grant
CCF-1947546.
}
\bigskip 

\includegraphics[height=37mm]{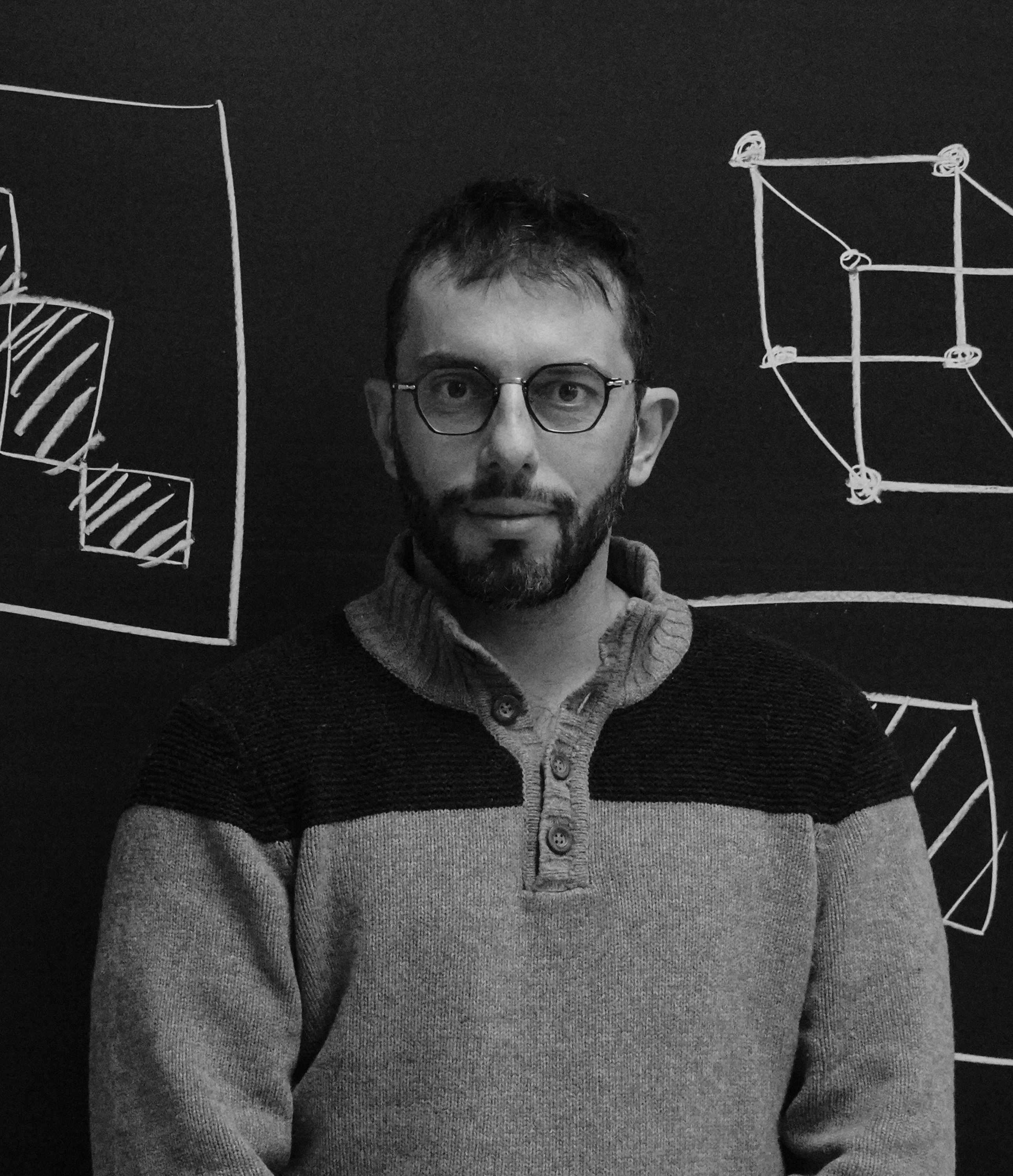}\hspace{.1in}\includegraphics[height=37mm]{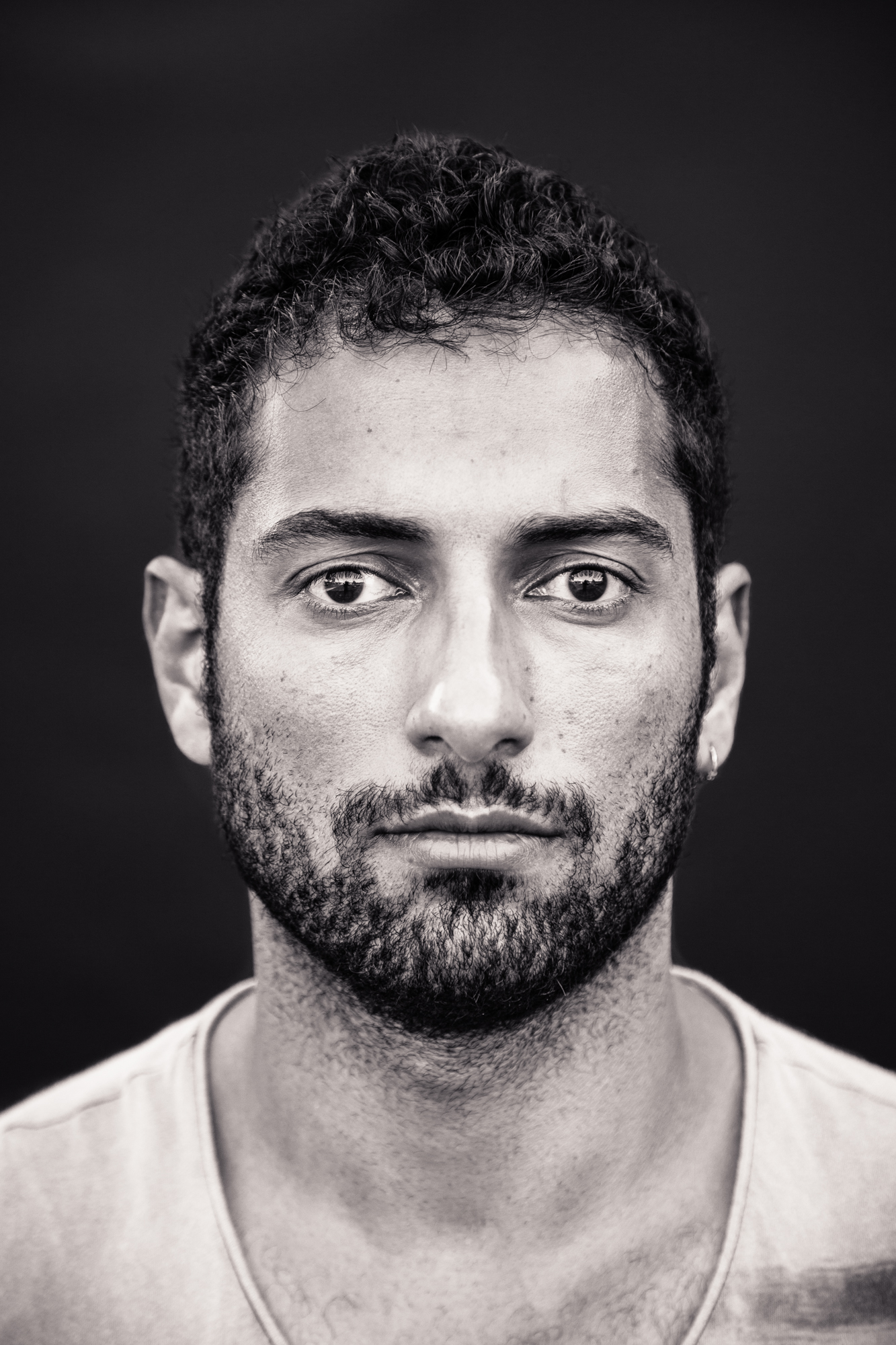}
\end{center}

\begin{abstract}
Several theorems and conjectures in communication complexity state or speculate that the complexity of a matrix in a given communication model is controlled by a related analytic or algebraic matrix parameter, e.g., rank, sign-rank, discrepancy, etc. The forward direction is typically easy as the structural implications of small complexity often imply a bound on some matrix parameter. The challenge lies in establishing the reverse direction, which requires understanding the structure of Boolean matrices for which a given matrix parameter is small or large. We will discuss several research directions that align with this overarching theme.

\vskip5pt 
\begin{updatebox} 
This survey was updated on July 29, 2026 to describe subsequent progress on its open problems. Material added in this update is enclosed in boxes.
\end{updatebox}
\end{abstract}

\section{Introduction}
In 1979, Yao~\cite{yao1979some} introduced an abstract model for analyzing communication. It quickly became apparent that the applications of this elegant paradigm go far beyond the concept of communication. Many results in communication complexity have equivalent formulations in other fields that are equally natural, and the techniques developed within this framework have proven to be powerful tools applicable across various domains.  Today, communication complexity is a vibrant research area with many connections across theoretical computer science and mathematics: in learning theory, circuit design, pseudorandomness, data streaming, data structures, computational complexity, computer networks, time-space trade-offs, discrepancy theory, and property testing. 

In this article, we focus on the most standard framework where a communication problem is simply a \emph{Boolean} matrix. Formally, there are two parties, often called Alice and Bob, and a communication problem is defined by a matrix  $F \in \{0,1\}^{\cX \times \cY}$.  Alice receives a row index $x \in \cX$, and Bob receives a column index $y \in \cY$. Together, they should compute the entry $F(x, y)$ by exchanging bits of information according to a previously agreed-on protocol tailored to $F$. There is no restriction on their computational power; the only measure we care to minimize is the number of exchanged bits.

Many questions in communication complexity concern the basic structural properties of Boolean matrices and are relevant to any field that requires an in-depth analysis of them. Mathematicians, of course, have studied matrices for centuries through the lenses of linear algebra, geometry, and analysis. They have produced an extensive collection of tools and theories that apply to any field dealing with these mathematical objects. However, the assumption of Booleanity introduces a novel angle and unveils new problems and challenges.

Consider the example of rank. Elementary linear algebra provides several satisfactory structural descriptions of small-rank matrices. For example, one could construct a small-rank matrix by summing a few rank-one matrices. In contrast, the structure of small-rank \emph{Boolean matrices} is the subject of the most well-known conjecture in communication complexity, \emph{the log-rank conjecture}.

There are many theorems and conjectures in communication complexity that fall into a similar paradigm as the log-rank conjecture: They state or speculate that the complexity of a matrix in a given communication model is essentially determined by a related analytic or algebraic matrix parameter, e.g., rank, sign-rank, discrepancy, trace norm, approximate trace norm, $\gamma_2$-factorization norm, approximate factorization norm.  The forward direction is typically easy as the structural implications of small complexity often imply a  bound on some matrix parameter. The challenge lies in establishing the reverse direction, which requires understanding \emph{the structure of Boolean matrices} for which a given matrix parameter is small or large.

In this article, we discuss some new research directions and open problems, as well as some classical ones, that align with this overarching theme. 

\paragraph{Notation:} All logarithms are in base $2$. Sometimes, we use $a \lesssim b$ to denote $a=O(b)$. Let $\I_m$  denote the $m \times m$ identity matrix and $\J_{\cX \times \cY}$ denote the $\cX \times \cY$ all-$1$ matrix. For a positive integer $k$, we denote $[k]= \{1,\ldots,k\}$. We often identify a Boolean matrix $F_{\cX \times \cY}$ with the corresponding function $F:\cX \times \cY \to \{0,1\}$ defined as $F:(x,y) \mapsto F(x,y)$. We define the \emph{complement} of a communication problem $F=F_{\cX \times \cY}$ as $\neg F = \J_{\cX \times \cY}-F$. For two $\cX\times \cY$ matrices $A$ and $B$, we denote their entry-wise product (i.e. \emph{Schur product}) by $A\circ B$.

\paragraph{Communication complexity classes:}

We measure the communication complexity of a matrix $F_{\cX \times \cY}$ in relation to the number of input bits $n(F) \defeq \lceil \log \max(|\cX|,|\cY|) \rceil$. We often consider $2^n \times 2^n$ matrices $F_n:\{0,1\}^n \times \{0,1\}^n \to \{0,1\}$, where the inputs of Alice and Bob are $n$-bit strings. Moreover, similar to computational complexity,  the goal is to understand the asymptotic complexity, and therefore, a communication problem typically refers to an infinite family of Boolean matrices rather than a single matrix.  

For example, \textsc{Equality} is the family of matrices $\EQ_n: \{0,1\}^n \times \{0,1\}^n \to \{0,1\}$ with $\EQ(x,y)=1$ iff $x=y$.  Equivalently, $\EQ_n$ is the $2^n \times 2^n$ identity matrix, where the rows and columns are labelled with $n$-bit strings.

In the theory of Turing machines, a polynomial complexity is considered efficient, but this is not a suitable criterion for communication since even in the deterministic model, communication complexity is at most $n+1$. In an influential paper, Babai, Frankl, and Simon \cite{bfs86} proposed polylogarithmic complexity as the criterion for efficiency and used it to define the communication classes $\Pcc^\cc,\NP^\cc,\RP^\cc,\BPP^\cc,\PPcc^\cc,\UPP^\cc$, analogous to the classical computational complexity classes. The definition of communication classes by~\cite{bfs86} provides a formal paradigm to compare the power of different communication models. 

Since in this article, we only study communication classes, to make the notation less cumbersome, from this point on, we will drop the superscript $\cc$ and denote these communication classes simply as $\Pcc,\NP,\RP,\BPP,\PPcc,\UPP$.

\paragraph{Constant-cost communication classes:} For many communication problems, it may not be completely natural to allow the communication complexity to depend on the matrix size. For example, consider the geometric problem where Alice receives a point $x \in \R^2$, and Bob receives a closed half-space $H \subseteq\R^2$, and they wish to know whether $x \in H$. Alternatively, consider the problem where Alice and Bob receive points $x,y \in \R^2$, respectively, and they wish to know whether the distance between these points is at most $1$. In these examples, the number of possible inputs is infinite, and even if we artificially restrict the number of possible inputs to a finite set,  the most natural way to represent the inputs is by vectors in $\R^2$.  Therefore, it is more natural to consider communication protocols that solve the problem (in a certain communication model or under some promised guarantees about the inputs) using a bounded number of communicated bits, independent of the number of possible inputs.

Furthermore, many parameters studied in communication complexity are also relevant to other areas of theoretical computer science, where the Boolean matrices are either infinite or their size is of lesser significance than other features. For example, in learning theory, notions such as VC dimension, Littlestone dimension, margin complexity, and Threshold dimension of binary concept classes are of primary interest rather than the class size. Viewing binary concept classes as Boolean matrices, these parameters are directly related to notions such as stability, discrepancy, and sign-rank that play a central role in communication complexity. 

The definition of communication classes by~\cite{bfs86}, no doubt, is natural and has successfully led to many fruitful lines of research seeking to prove separations between different communication classes. However, if one wishes to view communication complexity in these broader contexts, it becomes essential also to analyze communication problems that have uniformly bounded complexity. Indeed, there is already a rapidly growing body of work dedicated to studying constant-cost communication~\cite{MR2359826,linialshraibman,MR3218542,MR4129280,MR4048136,HWZ22,hatami2022lower,esperet2022sketching,HHH23,CHHS23,ahmed2023communication,harms2023randomized}. 

These considerations motivate the definition of the \emph{constant-cost} analogues of the communication classes of~\cite{bfs86}. Here, the criterion of effectiveness is an $O(1)$ complexity, independent of the input size $n$. We will denote these classes by adding a $0$ subscript to the corresponding polylogarithmic class, where $0$ refers to the power $d$ in the polylogarithmic complexity $O(\log^d(n))$. For example, a family of Boolean matrices $F_n$ is in $\Pcc_0$ if their deterministic communication complexity is uniformly bounded by a constant $C=O(1)$. 

We will formally define and discuss the classes $\Pcc_0,\NP_0,\RP_0,\BPP_0,\PPcc_0,\UPP_0$ later, but for now, let us mention that, unlike in the polylogarithmic communication classes where $\Pcc \subsetneq \RP  \subsetneq \NP$ and $\BPP \subsetneq \PPcc$, here we have 
$\Pcc_0=\NP_0 \subsetneq \RP_0$ and $\BPP_0=\PPcc_0$. We also believe that $\BPP_0 \not \subseteq \UPP_0$, while Newman's lemma~\cite{newman} implies that  $\BPP  \subseteq \UPP$.

\paragraph{Monochromatic rectangles:} 
A \emph{monochromatic rectangle} in a Boolean matrix  $F_{\cX\times \cY}$ refers to a submatrix $F_{S \times T}$, where all entries share the same value. Monochromatic rectangles are the building blocks of every communication protocol, and many results and problems in communication complexity focus on describing the structure of Boolean matrices in terms of these rectangles or showing the existence of large monochromatic rectangles in them. 

Define the \emph{rectangle ratio}  of a Boolean matrix $F_{\cX \times \cY}$ as \[ \rect(F) \defeq \max_{R} \frac{|R|}{|\cX \times \cY|},\] where the maximum is over monochromatic rectangles in $F$.

Concerning communication complexity, the parameter $\rect(\cdot)$ has a few drawbacks: While adding new rows and columns to a matrix can only increase the communication complexity, it can deteriorate $\rect(F)$ by creating large monochromatic rectangles. Moreover, repeating a row or column of $F$ does not change its communication complexity, but it can affect $\rect(F)$. Therefore, it is more natural to consider the following version of $\rect$ that considers weighted monochromatic rectangles, first studied by Impagliazzo and Williams~\cite{impagliazzo2010communication}.

\begin{definition}[Weighted Rectangle Ratio]
\label{def:rectangleRatio}
For every Boolean matrix $F$, define 
\[\wrect(F) \defeq \inf_{\mu} \max_{R}\mu(R),  \]
where the infimum is over all \emph{product} probability measures $\mu=\mu_\cX \times \mu_\cY$ on $\cX \times \cY$, and the maximum is taken over all monochromatic rectangles in $F$.
\end{definition}

Both $\rect(\cdot)$ and $\wrect(\cdot)$ are natural parameters that measure the existence of structure in Boolean matrices. As we will discuss throughout the paper, they play a crucial role in several applications in communication complexity, and some basic and consequential questions regarding these quantities remain open.

\section{Deterministic communication complexity and rank}
\label{sec:deterministic}

A \emph{deterministic communication protocol} is defined by a binary tree, where every internal node $v$ specifies which player speaks at that node and what bit they must send. For example, if an internal node $v$ is associated with Alice, it is labelled with a function $a_v: \cX \to \{0, 1\}$, which prescribes the bit sent by Alice at this node. After this bit is sent, the players move to the corresponding child of $v$: they move to the left child if the bit is $0$ and to the right child if the bit is $1$.  They continue this process until they reach a leaf. Every leaf is labelled with $0$ or $1$, which corresponds to the output of the protocol.  The \emph{cost} of the protocol is the height of the tree, which is equal to the maximum number of bits exchanged on any input.  The \emph{deterministic communication complexity} of $F$, denoted by $\DD(F)$, is the smallest cost of a protocol that computes $F$ correctly on all inputs. Let $\Pcc$ and $\Pcc_0$ be, respectively, the class of problems with polylogarithmic and $O(1)$ deterministic communication complexity.

The fact that the bit communicated at every node depends only on one of the inputs implies the most useful fact in communication complexity: for every node $v$, the set of all inputs that lead the protocol to $v$ is a \emph{combinatorial rectangle}, which is a product set $R_v = S_v \times T_v$ with $S_v \subseteq \cX$ and $T_v \subseteq \cY$.

For every $b \in \{0,1\}$, a combinatorial rectangle $R$ is called a $b$-\emph{monochromatic rectangle} if $F(x,y)=b$ for every entry $(x,y) \in R$.  Note that for every $b$-labelled leaf $\ell$, the set $R_\ell$ must be a $b$-monochromatic rectangle. Therefore, the leaves of a deterministic protocol partition the matrix into monochromatic rectangles. 

A Boolean matrix $F$ with small $\DD(F)$ is highly structured, as it consists of at most $2^{\DD(F)}$ monochromatic rectangles.  This partition of $F$ into at most $2^{\DD(F)}$ disjoint monochromatic rectangles implies
\begin{equation}
\DD(F) \ge \log(1/\rect(F)).
\end{equation}
Moreover, since the real rank of a monochromatic rectangle is at most $1$, we have $\rank(F) \le  2^{\DD(F)}$, where $\rank(F)$ denotes the real rank of $F$. Combined with the easy fact $\DD(F) \le  \rank(F)+1$, we have
\begin{equation}
\label{eq:rank_lower_bound}
\log \rank(F) \le \DD(F) \le \rank(F)+1.
\end{equation}
Therefore, $\Pcc_0$ coincides with the class of problems with bounded rank. 
However, the exponential gap between the lower bound and the upper bound is too wide to provide a characterization of $\Pcc$ in terms of rank. The \emph{log-rank conjecture} speculates that the lower bound is essentially sharp and $\DD(F)$ and $\log\rank(F)$ are polynomially equivalent. 

\begin{conjecture}[The log-rank conjecture~\cite{lovasz1988lattices}]
\label{conj:LRC}
There exists a universal constant $C > 0$ such that for
every Boolean matrix $F$,
\[\DD(F) \le C \left(\log \rank(F)\right)^C.\]
\end{conjecture}

The log-rank conjecture, if true, provides a structural description for Boolean matrices of small rank. Namely, they must have the same tree-like partition into monochromatic rectangles as matrices with small deterministic communication complexity. Nisan and Wigderson~\cite{nisan1995rank} showed that to prove the log-rank conjecture, it suffices to show the existence of one large monochromatic rectangle. 

\begin{conjecture}[The log-rank conjecture - equivalent formulation]
\label{conj:LRC-eq}
There exists a universal constant $C > 0$ such that every Boolean matrix $F$ satisfies 
\[ \rect^{-1}(F) \leq  2^{C \log^C \rank(F)}.\]
\end{conjecture}

To date, the strongest known bound is still exponentially far from that in \cref{conj:LRC}. Concretely, Lovett \cite{lovett2016communication} showed that $\DD(f) \le O(\sqrt{\rank(f)} \log \rank(f))$, which was recently improved by Sudakov and Tomon~\cite{sudakov2023matrix} to $O(\sqrt{\rank(f)})$. On the lower bound side, a construction of \cite{goos2018deterministic} shows that the constant $C$ cannot be strictly less than $2$. We refer to surveys~\cite{MR3222332,lee2023around} for a detailed discussion of the log-rank conjecture.

\section{Equality oracles and $\gamma_2$-Factorization norm}

The results discussed in this section mostly stem from joint work by Lianna Hambardzumyan and the authors~\cite{HHH23}. Regarding communication complexity, we will focus on the deterministic model with access to an {\sc Equality} oracle. We will show that this model has compelling ties to harmonic analysis and operator theory. In particular, we will discuss a conjecture in communication complexity with profound links to the characterization of the idempotents of the algebra of Schur multipliers and Cohen’s celebrated idempotent theorem, a well-known and notable theorem in harmonic analysis.
 
\subsection{An analytic version of the rank lower-bound}\label{sec:analyticrank}

The rank lower bound on $\DD(F)$ follows from the observation that every Boolean matrix $F$ can be written as  $F = \sum_{i=1}^{2^{\DD(F)}} B_i$, where each $B_i$ is a rank-one \emph{Boolean} matrix. If instead of counting the number of rank-one \emph{Boolean} matrices in this sum, we focus on the sum of the coefficients, we can establish an alternative lower bound on $\DD(F)$—a lower bound with an analytical flavour.

Define the $\mu$-norm of a \emph{real} matrix $A$  as 
\[ \norm{A}_\mu= \inf_{\lambda_i, B_i} \left\{ \sum_i |\lambda_i| \ : \ A=\sum_i \lambda_i B_i \right\},\]
where each $B_i$ is a rank-one \emph{Boolean} matrix and each $\lambda_i$ is a real number. Note that $\norm{\cdot}_\mu$ satisfies all the axioms of a norm, and we have
\begin{equation}
\label{eq:mu}
    \log \norm{F}_\mu \le   \DD(F).
\end{equation}

The $\mu$-norm is not well-known in functional analysis, and its definition is tailored to its purpose as a lower bound in communication complexity. Fortunately, the Grothendieck inequality shows that the $\mu$-norm is equivalent to the well-studied $\gamma_2$-factorization norm. 

\begin{definition}
\label{def:gamma_2}
The $\gamma_2$-norm of a real matrix $A_{\cX \times \cY}$, denoted by $\norm{A}_{\gamma_2}$, is the infimum of $c$ such that there exists a positive integer $d$ and vectors $u_x,v_y \in \R^d$ for $x \in \cX$ and $y \in \cY$  with $\inp{u_x}{v_y}=A(x,y)$ and $\norm{u_x}_2 \cdot \norm{v_y}_2 \le c$ for all $x,y$.
\end{definition}

A key property of the $\gamma_2$-norm is that for every two  $\cX \times \cY$ real matrices $A_1$ and $A_2$, we have 
\begin{equation}
\label{eq:algebra}
\norm{A_1 \circ A_2}_{\gamma_2} \le \norm{A_1}_{\gamma_2} \norm{A_2}_{\gamma_2},
\end{equation}
where we recall that $A_1 \circ A_2$ is the entry-wise product of  $A_1$ and $A_2$.

The Grothendieck inequality implies that for every real matrix $A$, we have 
\[\norm{A}_{\gamma_2} \le \norm{A}_\mu \le 4 K \norm{A}_{\gamma_2}, \]
where $K \le \frac{\pi}{2\ln(1+\sqrt{2})} \approx 1.7822$ is the so-called Grothendieck constant for real numbers.  In light of this equivalence, we will rephrase the lower bound \cref{eq:mu} in terms of the $\gamma_2$-norm: 
\begin{equation}
\label{eq:gamma_2}
    \log \norm{F}_{\gamma_2} \leq   \DD(F). 
\end{equation}

If we compare this lower bound to \cref{eq:rank_lower_bound},  it is natural to wonder whether we can bound $\DD(F)$ from above by a function of $\norm{F}_{\gamma_2}$. The answer is negative. The $\gamma_2$-norm of every identity matrix is $1$, but identity matrices can be of arbitrarily large rank and, therefore, of arbitrarily large deterministic communication complexity. 

\begin{proposition}
\label{prop:identity}
For every $m$, the  $\gamma_2$-norm of the $m\times m$ identity matrix $\I_m$ is $1$,
\end{proposition}
\begin{proof}
Note that since the standard basis $e_1,\ldots,e_m \in \R^m$ satisfies 
\[
\inp{e_i}{e_j}=
 \begin{cases}
    1 & i=j \\
    0 & i \neq j 
  \end{cases},
\]
by \cref{def:gamma_2}, $\I_m$ satisfies $\norm{\I_m}_{\gamma_2} \le 1$. Moreover, it is clear from the definition of the $\gamma_2$-norm that for every matrix $A$, we have $\norm{A}_{\gamma_2} \ge \norm{A}_\infty \defeq \max_{x,y} |A(x,y)|$, and therefore, $\norm{\I_m}_{\gamma_2} \ge 1$. 
\end{proof}

\paragraph{What are the Boolean matrices with small $\gamma_2$-norm  then?} 
It follows from $\norm{A}_{\gamma_2} \ge \norm{A}_{\infty}$ that the $\gamma_2$-norm of every non-zero Boolean matrix is at least $1$. Let us first study the Boolean matrices whose $\gamma_2$-norm is exactly $1$. 

\cref{prop:identity} shows that all identity matrices have $\gamma_2$-norm $1$. Note that the proof of \cref{prop:identity} generalizes to a larger class of matrices, which we call \emph{blocky matrices}.

\begin{definition}[Blocky matrices]
A Boolean matrix $F_{\cX \times \cY}$  is  \emph{blocky} if there exist disjoint sets $\cX_i \subseteq \cX$ and disjoint sets $\cY_i  \subseteq \cY$  such that  the support of $F$ is exactly $ \bigcup_{i} \cX_i \times \cY_i$. 
\end{definition}

Let $\Blocky$ denote the set of all blocky matrices. It turns out that $\Blocky$ is precisely the set of Boolean matrices with $\gamma_2$-norm $1$. 

\begin{proposition}[Livshits~\cite{MR1332920}]
\label{prop:Liv}
A Boolean matrix $F$ satisfies $\norm{F}_{\gamma_2}=1$ iff $F \in \Blocky$. 
\end{proposition}
\begin{proof}
It is observed in~\cite{MR1332920} that 
\[\left\| 
\begin{bmatrix}
	1 & 1 \\
	0 & 1 \\
\end{bmatrix} \right\|_{\gamma_2} = \frac{2}{\sqrt{3}}>1. \]
Since  $\norm{\cdot}_{\gamma_2}$ norm is invariant under row and column permutations, a Boolean matrix $F$ with $\norm{F}_{\gamma_2}=1$ cannot have any $2 \times 2$ submatrices with exactly $3$ ones.  It is straightforward to verify that a Boolean matrix satisfying this property must be blocky. 
\end{proof}

Since every Boolean matrix with $\gamma_2$-norm $1$ is blocky, it is natural to ask whether Boolean matrices of bounded $\gamma_2$-norm can be characterized through blocky matrices.

We first show that if $F$ is generated by entry-wise operations from a few blocky matrices, it must have a small $\gamma_2$-norm. 

\begin{proposition}
\label{prop:FunctionalBlocky}
Consider $\cX\times \cY$ blocky matrices $B_1,\dots, B_r$ and a combining function $\Gamma:\zo^r\rightarrow \zo$. The  Boolean matrix $F_{\cX \times \cY}$ defined as  
\begin{equation}\label{eq:blockyfuncrank}
F(x,y) \defeq \Gamma(B_1(x,y),\dots, B_r(x,y))
\end{equation}
satisfies  $\norm{F}_{\gamma_2} \le 3^r$.  
\end{proposition}
\begin{proof}
We prove the statement by induction on $r$. The base case $r=0$ is trivial. Next, consider a Boolean matrix $F$ satisfying \cref{eq:blockyfuncrank}, and write $F = (B_1 \circ F_1) + (\J - B_1) \circ F_2$, where the entries of $F_1$ and $F_2$ depend only on $B_2,\ldots, B_{r}$ and $\J$ is the all-$1$ matrix. Note that $\J$ is a blocky matrix and satisfies $\norm{\J}_{\gamma_2}=1$. We have 
\begin{align*}
\norm{F}_{\gamma_2} &\le  \norm{B_1 \circ F_1}_{\gamma_2}+ \norm{(\J-B_1) \circ F_2}_{\gamma_2} \\ &\le \norm{F_1}_{\gamma_2} + (\norm{\J}_{\gamma_2} + \norm{B_1}_{\gamma_2})  \norm{F_2}_{\gamma_2} \le 3^{r-1}+2 \cdot 3^{r-1}=3^r.\qedhere
\end{align*}
\end{proof}

In \cite{HHH23}, we conjectured that Boolean matrices of small $\gamma_2$-norm are precisely those of the form \cref{eq:blockyfuncrank}.

\begin{conjecture}[\textcolor{red}{Update: Proved in \cite{goh2026characterizationidempotentschurmultipliers}}]
\label{conj:Blocky}
Suppose that $F$ is a Boolean matrix with $\norm{F}_{\gamma_2} \le c$. Then we
may write
\begin{equation}
F = \sum_{i=1}^L \pm B_i,  
\end{equation}
where $B_i$ are blocky matrices and $L \le \ell(c)$ for some integer $\ell(c)$ depending only on $c$.  
\end{conjecture}

\begin{updatebox}
\begin{itemize}
\item  Balla, Hambardzumyan, and Tomon~\cite{bht26} proved that every Boolean matrix with  $\norm{F}_{\gamma_2} \le c$ contains a monochromatic rectangle of density $2^{-O(c^3)}$. 
\vskip5pt

\item  Goh and Hatami~\cite{blockysubset2025} proved that the support of such a matrix contains a blocky matrix covering a significant portion of its $1$-entries. In~\cite{MR5000800}, the same authors proved that for every such $n \times n$ matrix, one may take $L \le 2^{O(c^7)} \log^2n$ in \Cref{conj:Blocky}.
\vskip5pt 

\item Beke, Goh, Hatami, Jaffe, and Naylor \cite{goh2026characterizationidempotentschurmultipliers}  finally established the conjecture, proving the bound $L \le 2^{O(c^6)}$. 
\end{itemize}
\end{updatebox}

\cref{conj:Blocky} is inspired by Cohen's idempotent theorem, and it is known to be true for a large class of Boolean matrices, including the {\sc xor}-lifts of Boolean functions. 

Recall that the   {\sc xor}-lift of a function $f:\Z_2^n \to \{0,1\}$ is the matrix $f^\oplus:\Z_2^n \times \Z_2^n \to \{0,1\}$ defined as $f^\oplus(x,y)=f(x + y)$. 

The sum of the absolute values of the Fourier coefficients of a function $f:\F_2^n \to \R$ is called the \emph{Fourier algebra norm} or \emph{spectral norm} of $f$ and is denoted by
\[ \norm{f}_A \defeq \norm{\widehat{f}}_1 = \sum_{\chi \in \widehat{G}} |\widehat{f}(\chi)|.\] 
The following identity relating the Fourier algebra norm of $f$ to the $\gamma_2$-norm of its {\sc xor} lift is due to~\cite[Lemma 36]{linial2009lower}.
\begin{equation}\label{eq:Avsgamma2}
\norm{f}_A  = \norm{f^\oplus}_{\gamma_2}. 
\end{equation}
Therefore, in the case of {\sc xor}-lifts, the assumption $\norm{f^\oplus}_{\gamma_2} \le c$ of \cref{conj:Blocky} is equivalent $\norm{f}_{A} \le c$. The structure of Boolean functions $f:\Z_2^n \to \{0,1\}$ with $\norm{f}_A \le c$ is characterized by the quantitative version of the so-called Cohen's idempotent theorem for $\Z_2^n$. 
 
\begin{theorem}[Quantitative Cohen's theorem for $\Z_2^n$ \cite{MR133397,Green_Sanders}] \label{thm:Cohen}
If $S \subseteq  \Z_2^n$ satisfies $\norm{\1_S}_{A} \le c$, we may write 
\begin{equation}
\1_S = \sum_{i=1}^L \pm \1_{H_i+a_i},  
\end{equation}
where $H_i+a_i$ are cosets and $L \le \ell(c)$ for some integer $\ell(c)$ depending only on $c$.  
\end{theorem}
 Note that the {\sc xor}-lift of every $\1_{H_i+a_i}$ is a blocky matrix, and therefore,  \cref{thm:Cohen} verifies \cref{conj:Blocky} for  {\sc xor}-lifts. In fact, these results extend to every finite group. Given a finite group $G$, we can generalize the notion of {\sc xor}-lifts to \emph{group}-lifts, where we define $F(x,y)=f(y^{-1}x)$ for $f:G\rightarrow \C$. The notion of algebra norm also generalizes to other finite groups.  We will not give the original definition of the algebra norm, but for the purposes of this paper, it suffices to know that similar to \cref{eq:Avsgamma2}, we have  $\norm{f}_A=\norm{F}_{\gamma_2} $. This identity was observed in \cite{HHH23} based on a result of Davidson and Donsig~\cite{MR2379721}.

\begin{theorem}[Quantitative Cohen's theorem for finite groups~\cite{MR2773105}]
\label{thm:Cayley}
There exists a function $\ell:\N \to \N$ such that the following holds. For every finite group $G$ and every $S \subseteq G$, if $F_{G \times G}(x,y)=\1_S(y^{-1}x)$ satisfies  $\norm{F}_{\gamma_2} \le c$, we may write 
\begin{equation}
\1_S = \sum_{i=1}^L \pm \1_{H_ia_i},  
\end{equation}
where $H_ia_i$ are cosets and $L \le \ell(c)$. In particular, 
$ F = \sum_{i=1}^L \pm B_i$, 
where $B_i$ are blocky matrices.
\end{theorem}

\cref{{thm:Cayley}} verifies \cref{conj:Blocky} for the adjacency matrix of every Cayley (directed) graph.  On the other hand, for general Boolean matrices, it is not even known whether $\norm{F}_{\gamma_2} \le c$  implies $\rect(F) \ge \kappa(c)$ for some $\kappa(c)>0$, which would be an easy consequence of \cref{conj:Blocky}. 

\begin{updatebox}
This was first shown by Balla, Hambardzumyan,
and Tomon~\cite{bht26}, who established the bound $\kappa(c)\ge 2^{-O(c^3)}$.
\end{updatebox}

Regarding the bound in \cref{thm:Cayley}, for general finite groups, one can take $\ell(c) = A(6, O(c))$, where $A$ is the Ackermann function~\cite{MR2773105}. For the case of finite Abelian groups, a better bound of $\ell(c)= 2^{O(c^4 \polylog(c))}$ is due to~\cite{sanders2020bounds}.

In the special case of $\Z_2^n$ which corresponds to the {\sc xor}-lifts, the best bound~\cite{sanders_2019} that appears in the literature is $\ell(c) \le 2^{O(c^3 \polylog(c))}$. However, recently, Gowers, Green, Manners, and Tao~\cite{gowers2023conjecture} announced a proof for Morton's conjecture (aka polynomial Freiman–Ruzsa conjecture). \textcolor{red}{[\textbf{Erratum:}} Substituting this result in Sanders' proof for \cite[Proposition 2]{sanders_2019} shows that in the case of $\Z_2^n$, one may take  $\ell(c)= 2^{O(c \polylog(c))}$. \textcolor{red}{This claim is incorrect. The best known bound is still $\ell(c) \le 2^{O(c^3 \polylog(c))}$: as explained in the last paragraphs of Section 2 of \cite{sanders_2019}, the constant $3$ arises at two different points in Sanders' proof, and the resolution of Morton's conjecture only addresses one of these two points.]}

\subsection{Equality Oracle Protocols} 
\textsc{Equality} is the canonical problem with the strongest possible separation between deterministic and randomized communication complexities. We have $\DD(\EQ_n)=n+1$, which is the largest possible value for any $n$-bit communication problem. On the other hand, as we will discuss in \cref{sec:BPP},  the randomized communication complexity of $\EQ_n$ is only $O(1)$.   

We know that the deterministic model cannot solve \textsc{Equality} efficiently. What if we augment the model with an equality oracle?  Does this result in a significantly stronger model? Can this model efficiently solve every problem with small randomized communication complexity? 

Formally, in the deterministic communication model with access to an \textsc{Equality} oracle, a protocol for a Boolean matrix $F_{\cX \times \cY}$ corresponds to a binary tree. Each non-leaf node $v$ in the tree is labelled with two functions $a_v:\cX \to \{0,1\}^m$ and $b_v:\cY \to \{0,1\}^m$ for some $m$.   On this node, the players map their inputs to strings $a_v(x)$ and $b_v(y)$, respectively,  and the oracle will broadcast the value of $\EQ_m(a_v(x), b_v(y))$  to both players. This will contribute only $1$ to the cost of the protocol. Note that the oracle queries can simulate sending one-bit messages from each party to the
other one. For example, if it is Alice's turn to send a bit $a$, the query $\EQ_1(a,1)$ can transmit it to Bob. Hence, in this model, we can assume that all the communication is through oracle queries.  
 
Let $\DD^{\EQ}(F)$ denote the smallest cost of a deterministic protocol with equality oracle for the matrix $F$, and define $\Pcc^\EQ$ and $\Pcc^\EQ_0$ to be, respectively, the class of problems with polylogarithmic and constant communication costs in this model.  

 In the same way that combinatorial rectangles are the building blocks of deterministic communication protocols, blocky matrices serve as the foundational components of equality oracle protocols. Indeed, every node $v$ of an equality oracle protocol for computing $F(x,y)$ corresponds to $B_v(x,y)=\EQ(a_v(x),b_v(y))$ where $B_v$ is a blocky matrix.

\cref{eq:rank_lower_bound} characterizes $\Pcc_0$ as the set of problems with $O(1)$ rank. Can we obtain a similar characterization for $\Pcc^\EQ_0$ via the blocky matrices? To this end, let us define a notion of rank based on blocky matrices.
 
\begin{definition}[Blocky Rank]
The \emph{blocky rank} of a real matrix $A$, denoted $\blockyrank(A)$, is the smallest integer $r$ such that $A$ is a real linear combination of $r$ blocky matrices.
\end{definition}

Blocky rank has interesting connections to circuit and communication complexity theory~\cite{AmirBlocky,HHH23}. The following proposition shows analogous bounds to \cref{eq:gamma_2} on $\DD^{\EQ}(F)$, and implies that a matrix family $\{F_n\}$ is in $\Pcc^\EQ_0$ iff $\blockyrank(F_n)=O(1)$.

\begin{proposition}[\cite{HHH23}] 
\label{prop:EqOracle}
For every Boolean matrix $F_{\cX \times \cY}$, we have
\[ \frac{1}{2}\log \blockyrank(F)  \le \DD^\EQ(F) \le \blockyrank(F)\]
   and
\begin{equation}\label{eq:gamma2vseq}
\log \norm{F}_{\gamma_2} \leq 2\cdot \DD^{\EQ}(F).
\end{equation}
\end{proposition}	
\begin{proof}
We first prove $\DD^\EQ(F) \le \blockyrank(F)$. Let $k= \blockyrank(F)$.  We construct an $\EQ$-oracle protocol for $F$. In advance, Alice and Bob agree on a decomposition $F=\sum_{i=1}^k \lambda_i B_i$, where $B_i$ is a blocky matrix and $\lambda_i \in \R$ for $i\in [k]$. Since each blocky matrix $B_i$ corresponds to an $\EQ$ query, for an input $(x,y)$, Alice and Bob make $k$ queries to the oracle to determine  $F(x,y)$.

For the lower bounds, let $d=\DD^{\EQ}(F)$. Consider a leaf $\ell$ in the $\EQ$-oracle protocol tree computing $F$ and let $P_{\ell}$ denote the path of length $k_{\ell} \le d$ from the root to $\ell$. Note that each non-leaf node $v$ in the tree corresponds to a query to the equality oracle, and each such query corresponds to a blocky matrix $B_v$. Define
$B_{v}^1 = B_v$ and $B_{v}^0 = \neg B_v = \J_{\mathcal{X} \times \mathcal{Y}} - B_v$.

Suppose $P_{\ell}=v_1,v_2, \ldots, v_{k_\ell}, \ell$, and consider the matrix
\[F_{P_{\ell}} \coloneqq B_{v_1}^{\sigma_{v_1}} \circ B_{v_2}^{\sigma_{v_2}} \circ \ldots \circ B_{v_{k_{\ell}}}^{\sigma_{v_{k_\ell}}},\]
where $\sigma_{v_i} \in \{0,1\}$ and $\sigma_{v_i} = 1$ iff the edge $(v_{i-1},v_i)$ is labelled by $1$. Hence, after simplification,  $F_{P_\ell}$ can be written as a sum of at most $2^d$ summands with $\pm 1$ coefficients, where each summand is a Schur product of at most $k_\ell$ blocky matrices. Observe that the Schur product of two blocky matrices is a blocky matrix. Thus, $F_{P_\ell}$ is a sum of at most $2^d$ blocky matrices with $\pm 1$ coefficients.

Summing over all the leaves that are labelled by $1$, we get
$F = \sum_{\ell \text{ is a 1-leaf}} F_{P_{\ell}}$. As the number of leaves is bounded by $2^d$, and each $F_{P_\ell}$ is a $\pm 1$ linear combination of at most $2^d$ blocky matrices, we have $\blockyrank(F) \leq 2^{2d}$ and $\norm{F}_{\gamma_2} \leq 2^{2d}$.
\end{proof}
 
\subsection{Analogue of the log-rank conjecture for blocky rank is false}

The log-rank conjecture speculates that the deterministic communication complexity is polynomially equivalent to the logarithm of the rank. In light of \cref{prop:EqOracle} it is natural to ask a similar question for $\DD^\EQ$ and  $\log\blockyrank$. Arkadev Chattopadhyay\footnote{private communication, no pun intended!} observed that the recent counter-example to the so-called log-approximate-rank conjecture by  Chattopadhyay, Mande, and Sherif~\cite{approx} implies that the answer is negative. 

Recall that a node in a directed graph is called a sink if all of its adjacent edges are incoming. 
Define a function $\SINK_m: \{0, 1\}^{m \choose 2} \to \{0,1\}$  where the input of length ${m \choose 2}$ specifies the orientation of the edges of the complete graph on $m$ vertices. The function outputs $1$ if there is a vertex that is a sink in the given orientation of edges and $0$ otherwise.  

Fix $m$, and for $i\in [m]$, define $\psi_i: \{0,1\}^{m \choose 2} \to \{0,1\}$ to be the indicator function of whether $i$ is a sink in the orientation given by $x$. Note that 
\begin{equation}
\label{eq:sink}
\psi_i(x)=1 \Leftrightarrow x_{j,i}=1 \ \forall j \neq i,
\end{equation}
where $x_{j,i}=1$ indicates that the edge between $i$ and $j$ is oriented towards $i$. Since no orientation of the complete graph has more than one sink, we have 
\[ \SINK_m(x) = \sum_{i=1}^m \psi_i(x).\]

We will consider the family of {\sc xor}-lifts of sink functions. Recall that the {\sc xor}-lift of a function $f:\{0,1\}^n \to \R$ is $f^\oplus:\{0,1\}^n  \times \{0,1\}^n  \to \R$ with $f^\oplus(x,y) \defeq f(x \oplus y)$. We have  
\[\SINK_m^\oplus = \sum_{i=1}^m \psi_i^\oplus.\]
It follows from \cref{eq:sink} that each $\psi_i^\oplus$ is a blocky matrix, and therefore $\blockyrank(\SINK_m^\oplus) \le m$. On the other hand,  Chattopadhyay, Mande, and Sherif~\cite{approx} prove that the $\Rcc(\SINK_m^\oplus)=\Theta(m)$.  Since $\Rcc$ provides a lower bound on $\DD^\EQ$ (see \cref{eq:DEQvsR}), we obtain the following theorem. 

\begin{theorem}[Chattopadhyay, Mande, and Sherif~\cite{approx}]
For the family of  Boolean matrices $F_m={\SINK}_m^\oplus$, we have 
$
\DD^{\EQ}(F_m)= \widetilde{\Omega}(m)$ and $  \blockyrank(F_m) \le m$.  

\end{theorem}

\subsection{Blocky matrices and Idempotents of Schur Multipliers} 
Let $\cX$ and $\cY$ be countable sets,
and let $B(\cY,\cX)$ denote the space of bounded linear operators $A:\ell_2(\cY) \to \ell_2(\cX)$ endowed with the operator norm:
\[ \norm{A} = \sup_{x \in \ell_2(\cY): \ \norm{x}_2=1} \norm{Ax}_2.\]
A matrix $M_{\cX \times \cY}$ is called a \emph{Schur multiplier} if, for every $A \in B(\cY,\cX)$, we have $M \circ A \in B(\cY,\cX)$. In other words, $\norm{M \circ A}<\infty$ for every $A=A_{\cX \times \cY}$ with $\norm{A}<\infty$. Note that Schur multipliers form an algebra with addition and Schur product: If $M_1$ and $M_2$ are Schur multipliers, then $M_1+M_2$ and $M_1 \circ M_2$ are both Schur multipliers.

Every Schur multiplier $M$  defines a map  $B(\cY,\cX) \to B(\cY,\cX)$ via  $A \mapsto M \circ A$, which assigns an operator norm to it:  
\[\norm{M}_\m \coloneqq \norm{M}_{B(\cY,\cX) \to B(\cY,\cX)}= \sup_{\substack{A \in B(\cY,\cX) \\  \norm{A}=1}}  \norm{M \circ A}.\]
Note that $\norm{\cdot}_\m$ is an algebra norm as for every $M_1$ and $M_2$, we have
\[\norm{M_1 \circ M_2}_\m \le  \norm{M_1}_\m \norm{M_2}_\m.\]
In other words, the algebra of Schur multipliers endowed with the norm $\norm{\cdot}_\m$ is a Banach algebra. A classical result, due to Grothendieck, shows that the multiplier norm coincides with the $\gamma_2$-norm. 
\begin{proposition}[{See~\cite[Theorem 5.1]{MR1441076}}] 
For every matrix $A$, we have $\norm{A}_\m = \norm{A}_{\gamma_2}$.
\end{proposition}

An element $a$ of a Banach algebra is said to be an \emph{idempotent} (aka \emph{projection}) if $a^2=a$. The following question arises naturally. 
\begin{quotation}
What are the \emph{idempotents} of the algebra of Schur multipliers?
\end{quotation}

Every idempotent $F$ of this algebra must satisfy $F=F \circ F$ and, therefore, is a Boolean matrix. However, not every (infinite) Boolean matrix is a bounded Schur multiplier, as it is possible to have $\norm{F}_\m=\infty$ for a Boolean matrix $F$. \cref{prop:Liv} shows that blocky matrices are precisely the set of all   \emph{contractive} idempotents. In other words,  an idempotent Schur multiplier satisfies $\norm{F}_\m \le 1$ iff it is a blocky matrix. 

\begin{question}[\textcolor{red}{Update: Proved in \cite{goh2026characterizationidempotentschurmultipliers}}]
Are the idempotent Schur multipliers precisely those Boolean matrices that can be written as a $\pm 1$-linear combination of \emph{finitely} many contractive idempotents (equivalently,  blocky matrices)?
\end{question}

A simple compactness argument, as outlined in~\cite{HHH23}, shows that this problem is equivalent to \cref{conj:Blocky}. Therefore, a positive answer to \cref{conj:Blocky} would characterize idempotents of Schur multipliers, analogous to Cohen's~\cite{MR133397} characterization of the idempotents of the Fourier–Stieltjes algebra.

\section{Nondeterministic Model and $\Pcc^\NP$} 
In a nondeterministic protocol $\pi$ for a problem $F_{\cX \times \cY}$, the parties receive a shared advice string $a$ and use it in a standard deterministic protocol $\pi_a$. We say that a protocol computes $F$ if 
\[
F(x,y)= 1 \Leftrightarrow \exists a, \pi_a(x,y)=1. 
\]
The cost of the protocol is the bit-length of $a$ plus the maximum cost of $\pi_a(x,y)$ over all choices of $a,x,y$. The nondeterministic communication complexity of $F$, denoted by $\Ncc(F)$, is the minimum cost of such a protocol for $F$. A matrix family is in the class $\NP$ if they have polylogarithmic nondeterministic communication complexity. 

Unlike in the Turing-Machine complexity, in the communication framework, it is known that $\Pcc= \NP\cap \coNP$, which follows from $\DD(F)=O(\Ncc(F)\cdot \Ncc(\neg F))$; see~\cite[Theorem 2.11]{MR1426129}. However, nondeterministic protocols are provably more powerful than deterministic ones, as can be demonstrated by the important example of the set intersection problem.

The \textsc{Set-Int} problem, $\intr_n$, is defined by $\intr_n(x,y)=1$ if there exists a coordinate $i$ such that $x_i=y_i=1$. Since the players can use their nondeterminism to guess the intersecting coordinate $i$, we have $\Ncc(\intr_n)=O(\log n)$. However, it is easy to see that $\DD(\intr_n)=n+1$. In fact, 
 \cite{bfs86} already in the 1980s proved that \textsc{Set-Int} does not belong to $\Pcc^\EQ$.

\paragraph{The structural properties of $\NP$:} The nondeterministic communication complexity of a problem is fully captured by its \emph{monochromatic rectangle covering number.} Let $\mathrm{C}^1(F)$ denote the minimum number of $1$-monochromatic rectangles required to cover the $1$ entries of $F$. It is easy to see~\cite{MR1426129} that
\begin{equation}\label{eq:NvsCover}
\Ncc(F)=\log(\mathrm{C}^1(F))+ O(1).
\end{equation}
Combined with $\DD(F)=O(\mathrm{C}^1(F))$, we have
\begin{equation}\label{eq:DvsN}
\DD(F)\leq O(2^{\Ncc(F)}).
\end{equation}
Therefore, $\Pcc_0=\NP_0=\coNP_0$.  The following proposition shows that nondeterministic protocols, while more powerful than deterministic ones, satisfy the same quantitative bound on $\wrect(\cdot)$. 
\begin{proposition}
For every Boolean matrix $F$, 
\label{prop:rect_NP}
\[
\Ncc(F)\gtrsim \log(1/\wrect(F)).
\]
\end{proposition}
\begin{proof}
Let $F_{\cX\times \cY}$ be a Boolean matrix, and  $c=\mathrm{C}^1(F) = O(2^{\Ncc(F)})$. Let $\mu_\cX \times \mu_\cY$ be a product probability measure on $\cX\times \cY$, and let $S_1\times T_1, \dots, S_c\times T_c$ be a $1$-monochromatic rectangle covering of $F$. The case $c\leq 1$ is trivial, so assume $c>1$. 

If there exists $i$ with $\mu_\cX(S_i)\cdot \mu_\cY(T_i)\geq 1/4c^2$, then we are done. So, assume otherwise that for every $i$, we have $\mu_\cX(S_i)\cdot \mu_\cY(T_i)< 1/4c^2$. Let $I$ be the set of indices $i$ such that $\mu_\cX(S_i)<1/2c$. Note that, if $i\notin I$, then $\mu_\cY(T_i)< 1/2c$. Now define, $A=\cX \setminus \cup_{i\in I} S_i$ and $B=\cY \setminus \cup_{j\not\in I} T_i$. It is easy to see that $A\times B$ is a $0$-monochromatic rectangle of $F$ and $\mu_\cX(A)\cdot \mu_\cY(B)> 1/4$. 
\end{proof}
Impagliazzo and Williams~\cite{impagliazzo2010communication} extended the bound in \cref{prop:rect_NP} to the more powerful model of deterministic communication with access to $\NP$ oracles. Let us first define this model formally. 

An \emph{oracle communication protocol} for a communication problem $F$ is a protocol where each node $v$ is either a regular communication node or it is labelled with a triple $(P_v,a_v,b_v)$ where $P_v$ is a Boolean matrix, and $P_v(a_v(x),b_v(y))$ is used to decide whether to travel to the left or the right child of $v$. 

\paragraph{The complexity class $\Pcc^{\mathsf{NP}}$.} 
The $\mathsf{D}^{\mathsf{NP}}$ cost of an oracle communication protocol is the largest cost of a path from the root to a leaf, which is the sum of the communicated bits plus the sum of $\Ncc(P_v)$ for every $v$ on the path. 

Define $\mathsf{D}^{\mathsf{NP}}(F)$ to be the smallest $\DD^{\NP}$ cost of an oracle communication protocol for $F$. The complexity class $\Pcc^{\NP}$  is the class of problems $\{F_n\}$ with $\DD^{\mathsf{NP}}(F_n)=\polylog n$. 

Among the extensive list of complexity classes detailed in  G\"o\"os, Pitassi, and Watson's article~\cite{goos_landscape}, titled  ``\emph{the landscape of communication complexity classes}'', $\Pcc^{\NP}$ is the largest non-probabilistic class for which an explicit lower bound is known. For example, consider the inner product problem $\IP_n:\{0,1\}^n \times \{0,1\}^n \to \{0,1\}$ defined as $\IP_n(x,y)=x_1y_1+\dots+x_ny_n \mod 2$. A simple argument, based on dimension and orthogonality (see~\cite[Claim 1.17]{MR4312803}), shows that every monochromatic rectangle in $\IP_n$ is of size at most $2^n$, and therefore,  $\wrect(\IP_n)^{-1}  \ge 2^{\Omega(n)}$. The following theorem of \cite{impagliazzo2010communication} shows that $\DD^\NP(\IP_n) = \Omega(n)$.  

\begin{theorem}[Impagliazzo and Williams~\cite{impagliazzo2010communication}]
\label{thm:IW}
For every Boolean matrix $F$, we have 
\[\DD^{\NP}(F) \gtrsim \log\left( \wrect(F)^{-1}\right). \]
\end{theorem}
One might ask whether $\log\left( \wrect(F)^{-1}\right)$ and $\mathsf{D}^{\mathsf{NP}}(F)$ are polynomially equivalent. The answer is negative as~\cite{goos2019query} constructs an explicit family of Boolean matrices exhibiting a large gap between the two quantities.

\begin{theorem}[G\"o\"os, Kamath, Pitassi, and Watson~\cite{goos2019query}]
\label{thm:GapPNP}
There exists a sequence of $2^n \times 2^n$ Boolean matrices $F_n$ satisfying $\mathsf{D}^{\mathsf{NP}}(F_n) \ge n^{\Omega(1)}$  and  $\log\left( \wrect(F_n)^{-1}\right) \le \log^{O(1)}(n)$. 
\end{theorem}

\section{Probabilistic Communication Models}
Next, we discuss probabilistic communication protocols where the players can act in a randomized fashion. Randomness can be introduced in two different ways: private randomness and public randomness. 

In a \emph{private-coin randomized protocol}, each player has access to their own independent random bits and can use them to decide which bit to send next. More precisely, Alice and Bob have access to random strings $R_A$ and $R_B$, respectively. These two strings are chosen independently, each according to some probability distribution described by the protocol. The bit sent by Alice at a node $v$ is determined by a function $a_v$ of both $x$ and $R_A$. Similarly, the bits sent by Bob are determined by functions of $y$ and $R_B$.

In the \emph{public-coin} model, the players have access to a shared source of randomness. In other words, Alice and Bob both receive the same random string $R$. The public-coin model is stronger than the private-coin model as the former can simulate the latter by setting $R=(R_A,R_B)$. 

The cost of a randomized protocol is the maximum number of communicated bits over all inputs and all choices of random strings. A probabilistic protocol is allowed to make errors. It is common to consider three types of errors: 

\begin{itemize}
	\item {\bf Two-sided error ($\BPP$):}   For every $x,y$,  the probability that the protocol makes an error on $(x,y)$ is at most $\epsilon$ for some $\epsilon<1/2$. When $\epsilon$ is a fixed constant strictly less than $1/2$, the protocol is called a bounded-error protocol.  The particular choice of $\epsilon$ is unimportant as a simple error reduction shows that it affects the complexity by only a constant factor. Therefore, as it is common, we will fix the error parameter to $\epsilon=1/3$.  

	\item   {\bf One-sided error ($\RP$):} In this setting, the protocol can only make an error if $F(x,y)=1$. In other words,  for every $x,y$ with $F(x,y)=0$, the protocol must always correctly output $0$, but for every $x,y$ with $F(x,y)=1$, it might output a wrong answer with probability at most $\epsilon$ for some fixed $\epsilon <1$. We will fix the error parameter to $\epsilon=1/3$.

	\item  {\bf Zero-error ($\ZPP$):} In this case, the output of a  protocol is $0$, $1$, or $\perp$, where $\perp$ indicates a failure to compute $F(x,y)$. The protocol must never output $0$ or $1$ erroneously; however, on every input, it is allowed to output $\perp$ with probability at most $\frac{1}{2}$. 
\end{itemize} 

A classical result in communication complexity, called Newman's lemma, states that in the two-sided error, one-sided error, and zero-error settings, the following is true. The difference between public-coin and private-coin randomized communication complexities of any $n$-bit  communication problem is $O(\log(n))$.

Newman's lemma shows that when defining the polylogarithmic communication complexity classes $\BPP, \RP, \ZPP$, it is unimportant whether we use shared randomness or private randomness. However, to define constant-cost classes  $\BPP_0, \RP_0, \ZPP_0$, we need to make a choice. It turns out that in the setting of private-coin, all these classes collapse to $\Pcc_0$. Therefore, we shall define these classes in the public-coin model. 

Regarding the zero-error protocols, the following theorem shows that even in the public-coin model, the zero-error randomized communication complexity of a matrix $F$ is polynomially equivalent to $\DD(F)$ and in particular $\ZPP_0=\Pcc_0$.  

\begin{theorem}[{\cite[Theorem 2.1]{DAVIS2022106293}}]
\label{thm:zeroerror}
The public-coin zero-error randomized communication complexity of every Boolean matrix $F$ is at least $\Omega(\DD(F)^{1/4})$. 
\end{theorem}
It is an open problem whether this bound can be improved to $\Omega(\sqrt{\DD(F)})$, which, if true, would be sharp. We will not further discuss the zero-error model and refer the interested reader to~\cite{DAVIS2022106293} for further reading.

\subsection{The power of randomness: $\BPP$}
\label{sec:BPP}

The \emph{randomized communication complexity} of a Boolean matrix $F$, denoted by $\Rcc(F)$, is the minimum cost of a public-coin randomized protocol with two-sided error $\le 1/3$. Let $\BPP$ and $\BPP_0$ be, respectively, the class of problems with polylogarithmic and $O(1)$ randomized communication complexities.

Are probabilistic protocols more powerful than deterministic protocols? The example of \textsc{Equality} shows that randomness can provide a significant advantage.   To test whether $x \neq y$, Alice and Bob can use their shared randomness to jointly sample a random subset $S \subseteq \{0,1\}^n$ at no cost and then, by exchanging two bits of information, indicate to each other whether their inputs belong to $S$. If they see a disparity, they can conclude confidently that $x \neq y$. They can run this test twice, and if they do not detect $x \neq y$, they declare $x=y$. Note that the probability of error is $\le 1/4$. Therefore, $\Rcc(\EQ_n)=O(1)$, and $\textsc{Equality}\in \BPP_0$. In particular, we have the relations $\Pcc \subsetneq \BPP$ and $\BPP_0 \not\subseteq \Pcc$. Also note that $\Rcc(\EQ_n)=O(1)$ implies via standard error-reduction that
\begin{equation}\label{eq:DEQvsR}
\Rcc(F) \lesssim  \DD^\EQ(F) \log \DD^\EQ(F),
\end{equation}
which establishes  $\Pcc^\EQ \subseteq \BPP$ and  $\Pcc_0^\EQ \subseteq \BPP_0$.

Which problems have efficient randomized protocols? A substantial portion of the literature in communication complexity is dedicated to lower-bound techniques against randomized communication complexity, and many celebrated results establish such lower bounds for important concrete problems, such as Set Disjointness~\cite{MR1192778}, Gap Hamming Distance~\cite{MR3023795}, and Halfspace~\cite{sherstov2008halfspace}. While it is possible to write a voluminous book about the lower bounds against randomized communication complexity, we know very little about what is inside $\BPP$ and $\BPP_0$. In fact, until recently, it was not known whether there is any problem in $\BPP$ that is not in $\Pcc^\EQ$. Let us list some classical problems in $\BPP$. 

\begin{itemize}
\item  \textsc{Greater-Than} is the family of communication problems $\GT_n:[2^n] \times [2^n] \to \{0,1\}$ where $\GT_n(x,y)=1$ iff $x \le y$. It is known~\cite{Nis93,viola2015communication,ramamoorthy2015} that $\Rcc(\GT_n)=\Theta(\log(n))$, and therefore, 
\[\textsc{Greater-Than} \in \BPP \setminus \BPP_0.\]
\item \textsc{Hypercube} is the family of communication problems $\Q_n:\{0,1\}^n \times \{0,1\}^n \to \{0,1\}$ where $\Q_n(x,y)=1$ iff $x$ and $y$ differ in exactly one coordinate.  Given $x,y\in \{0,1\}^n$, Alice and Bob can pick a uniform partition of $[n]$ into $8$ sets $S_1,\ldots,S_8$ and accept if for exactly one $i\in [8]$, it holds that $(\oplus_{j \in S_i} x_j) \oplus (\oplus_{j \in S_i} y_j)=1$. It is easy to see that the communication cost of this protocol is constant and that the error probability is at most $1/3$.  Therefore, $\Rcc(\Q_n)=O(1)$. On the other hand, \cite[Lemma 2.15]{HHH23} and \Cref{eq:Avsgamma2} shows that $\norm{\Q_n}_{\gamma_2} \ge \Omega(\sqrt{n})$ and therefore, $\DD^{\EQ}(\Q_n) \ge \Omega(\log n)$. We have  \[\textsc{Hypercube} \in \BPP_0 \setminus \Pcc_0^\EQ;\]
see also \cite{HWZ22} for a different proof of this fact.%
\item More generally, let $\ell(n)<n/2$ be an integer, and $S_n \subseteq \{0,\ldots,\ell(n)\} \cup  \{n-\ell(n),\ldots,n\}$ and denote the hamming weight of an $x\in \{0,1\}^n$ by $|x|$. If $\ell(n)=\polylog(n)$, then the family of the \textsc{xor}-lifts  $\1_{S_n}^\oplus:\{0,1\}^n \times \{0,1\}^n \to \{0,1\}$ defined as $\1_{S_n}^\oplus(x,y)=\1_{S_n}(|x \oplus y|
)$ is in $\BPP$. If $\ell(n)=O(1)$, then this family is in $\BPP_0$~\cite{yao2003power}. Note that \textsc{Hypercube} corresponds to $S_n=\{1\}$.   
\item \textsc{Integer Inner product}:  Given a fixed positive integer $t$, the communication problem $\IIP_t$ is the family of functions $\IIP_{t,n}:[-2^n,2^n]^t \times  [-2^n,2^n]^t \to \{0,1\}$ with $\IIP_{t,n}([x_1,\ldots,x_t],[y_1,\ldots,y_t])=1$ iff $x_1y_1+\ldots+x_ty_t=0$. To check the validity of the equation, Alice and Bob can choose a random prime $p \approx \log(n)$ and exchange $x_i \mod p$ and $y_i \mod p$ for $i=1,\ldots,t$. This leads to a randomized protocol for $\IIP_t$ with cost $O(\log(n))$, which shows $\IIP_t \in \BPP$. On the other hand, for any fixed $t>2$, it was shown in \cite{chattopadhyay2019equality} that $\DD^\EQ(\IIP_{t,n}) \ge \Omega(n)$. In fact, as \cite{CHHS23} shows, one even has $\norm{\IIP_{t,n}}_{\gamma_2} \ge 2^{\Omega(n)}$. Hence, for $t>2$,  
\[\IIP_t \in \BPP \setminus \Pcc^\EQ.\]
\end{itemize}

We are unaware of any examples in $\BPP$ that fundamentally differ from those listed above. In fact, $\IIP_t$, which was introduced by Chattopadhyay, Lovett, and Vinyals~\cite{chattopadhyay2019equality}, is the only known example of a communication problem in $\BPP$ that is not in $\Pcc^\EQ$  (see also \cite{PSS23,CHHS23}). Let us mention a conjecture about $\IIP_t$ before proceeding further. We do not know how to prove any $\omega(1)$ lower bound for $\Rcc(\IIP_{t,n})$.  

\begin{conjecture}[{See \cite[Conjecture 6.4]{CHHS23} \textcolor{red}{Update: Solved in \cite{GoosHarmsRichterSofronova2026}}}]
\label{conj:IIP}
For $t>2$, \[\IIP_t \not\in \BPP_0.\] 
\end{conjecture}

Note that disproving \cref{conj:IIP} would imply that $\BPP_0 \not\subseteq \Pcc^{\EQ}$.

\begin{updatebox}

\begin{itemize}
\item  Consider the  point-line incidence communication problem where Alice receives $
u=(x_1,x_2)\in [-2^n,2^n]^2$, representing the line $y=x_1 x+x_2$, and Bob receives a point
$v=(y_1,y_2)\in [-2^n,2^n]^2$. Their goal is to decide whether the point lies on the line, or equivalently, whether
\[
x_1y_1+x_2-y_2=0.
\] 
We denote this communication problem by $\PL_{n}$. In a recent paper, G\"o\"os, Harms, Richter, and Sofronova~\cite{GoosHarmsRichterSofronova2026} gave an elegant application of the circle method to show that $\Rcc(\PL_{n}) = \Theta(\log n)$. In particular, it follows that for every fixed $t>2$,  $\Rcc(\IIP_{t,n}) = \Theta(\log n)$, and therefore, $\IIP_t \not\in \BPP_0$. 

\item   G\"o\"os, Harms, and Riazanov~\cite{GoosHarmsRiazanov2025} proved the existence of $n$-bit communication problems $F$ with $\Rcc(F) = O(1)$ versus $\DD^{\EQ}(F) \gtrsim \sqrt{n}$. In particular, they established $\BPP_0 \not\subseteq \Pcc^{\EQ}$.
\vskip5pt

\item  Goh and Hatami~\cite{gohhata2026} showed there is an arrangement of $2^n$ lines and $2^n$ points on the plane such that the corresponding $n$-bit point-line incidence communication problem $F$ gives the optimal separation $\Rcc(F) = O(1)$ versus $\DD^{\EQ}(F) \gtrsim n$. 
\end{itemize}
\end{updatebox}
The communication problems in $\Pcc^\EQ$ are highly structured as they are linear combinations of a few blocky matrices.  On the other hand, the only known example in $\BPP \setminus \Pcc^\EQ$ is the \textsc{integer inner product}, which has a low-dimensional geometric representation (i.e. bounded sign-rank) and enjoys nice structural properties. 

\begin{updatebox}
 The recent results above provide further examples in
$\mathrm{BPP}\setminus \Pcc^{\mathrm{EQ}}$, even in the constant-cost setting.
\end{updatebox}

All these known examples contain large monochromatic rectangles. Does every Boolean matrix with an efficient randomized communication protocol contain a large monochromatic rectangle? More specifically, G\"o\"os, Kamath, Pitassi, and Watson~\cite{goos2019query} asked the following question. 
\begin{question}
\label{q:BPP_large_rect}
Is it the case that for every family of $n$-bit communication problems $F_n$ in $\BPP$, there exists $c>0$ such that 
$\rect(F_n) \geq 2^{-c (\log n)^c}$?
\end{question}
By \cref{thm:IW}, a negative answer to \cref{q:BPP_large_rect} would imply that $\BPP\not\subseteq \Pcc^{\NP}$, a relation that remains unknown. 
\begin{conjecture}\label{conj:bppvsPNP}
$\BPP\not\subseteq \Pcc^{\NP}$. 
\end{conjecture}
Independently of~\cite{goos2019query}, and also motivated by  \cref{conj:bppvsPNP}, \cite{chattopadhyay2019equality} asked whether there exists a $c>0$ such that every communication problem $F$ satisfies $\rect(F)\geq 2^{-c\Rcc(F)^c}$. In fact, we do not know whether there is a uniform lower bound on $\rect(F)$ depending only on $\Rcc(F)$.

\begin{question}[{\cite{chattopadhyay2019equality,HHH23}}]\label{q:uniformrectangle}
Is there a function $\kappa:\N \rightarrow (0,1)$ such that 
$\rect(F)\geq \kappa(\Rcc(F))$?
\end{question}
 
As we shall see in \cref{thm:monochromaticSignRank}, a negative answer to \cref{q:uniformrectangle} would imply that $\BPP_0 \not\subseteq \UPP_0$, which is currently unknown.

\subsection{Is two-sided error necessary?}

Define the one-sided randomized communication complexity  $\Rcc^1(F)$  and its corresponding complexity classes $\RP$ and $\RP_0$ analogous to the two-sided error counterparts $\Rcc(F)$, $\BPP$, and $\BPP_0$.

\begin{quote}
Are two-sided error protocols genuinely more powerful than one-sided error protocols?
\end{quote}

One could give an affirmative answer to this question by referring to \textsc{Equality}, which satisfies $\Rcc(\EQ_n)=O(1)$ while $\Rcc^1(\EQ_n)=\Omega(n)$. This, however, is not a fully satisfactory separation. Indeed, since $\Rcc^1(\neg \EQ_n)=O(1)$, we can solve \textsc{Equality} with a single oracle query to the \textsc{Nonequality} problem which belongs to $\RP_0$. In other words, $\textsc{Equality} \in \coRP_0$.

If we examine all the known examples in $\BPP$, we realize that they all essentially boil down to solving problems with one-sided error in the sense that they are either in $\RP \cup \coRP$, or they are composed of a few components, each belonging to $\RP \cup \coRP$. 
We find this surprising, as we are not aware of any evident reasons as to why a two-sided error protocol might be simulated by a series of steps that can be performed by efficient one-sided error protocols. We suspect this phenomenon to be due to our limited knowledge of examples in $\BPP$. 
 
Define the class $\Pcc^{\RP}$ similarly to $\Pcc^{\NP}$, except that  the protocol is now charged $\Rcc^1(P_v)$ for its oracle queries $P_v$ at a node $v$. The simple inclusions of $\RP\subseteq \BPP$, $\Pcc^{\RP} \subseteq \BPP$ and  $\Pcc_0^{\RP} \subseteq \BPP_0$ are immediate from the definitions.  
\begin{question}\label{q:bppvsprp}
Is it true that  $\BPP= \Pcc^{\RP}$?
\end{question}

It is known that nondeterministic protocols can simulate one-sided protocols with a logarithmic loss,
\begin{equation}\label{eq:NvsR1}
\Ncc(F) \leq \Rcc^1(F)+O(\log n).
\end{equation} 
This shows that $\RP\subseteq \NP$ and $\Pcc^{\RP}\subseteq \Pcc^{\NP}$. In particular, \cref{conj:bppvsPNP} would imply a negative answer to \cref{q:bppvsprp}.

It is interesting to ask the above questions in the constant-cost setting. 

\begin{question}
Is it true that  $\BPP_0 = \Pcc_0^{\RP}$? 
\end{question}

\cref{thm:rectfromR1} below implies that for every communication problem in $\Pcc_0^{\RP}$, we have $\wrect(F)=\Omega(1)$. In particular,  a negative answer to \cref{q:uniformrectangle} would imply  $\BPP_0 \not\subseteq \Pcc_0^{\RP}$.

\begin{theorem}[{\cite[Theorem~3.8]{HHH23}}]
\label{thm:rectfromR1}
For every communication problem $F$, \[\wrect(F)\ge 2^{-O(\Rcc^1(F))}.\]
\end{theorem}

\subsection{Sign-rank and $\UPP$}

The \emph{unbounded-error communication complexity} of $F$, denoted by $\UU(F)$, is the smallest communication cost of a \emph{private-coin} randomized protocol $\pi$ that satisfies \[\Pr[\pi(x,y) \neq F(x,y)]< \frac{1}{2} \qquad \forall x,y. \]
In other words, the protocol is only required to outperform a random guess. The complexity classes corresponding to this measure are $\UPP$ and $\UPP_0$.

It is crucial that in this communication model, the players have only access to private randomness. Otherwise, given access to shared randomness, they could jointly sample a random input $(x_0,y_0)$ at no cost and use two bits of communication to verify whether $(x,y)=(x_0,y_0)$. If $(x,y)=(x_0,y_0)$, then they know the output $F(x,y)$, and if it is not, they can output a random bit. This protocol has an error probability strictly less than $1/2$. 

Paturi and Simon~\cite{paturi1986probabilistic} proved that the unbounded-error communication complexity is precisely determined by an elegant matrix parameter called \emph{sign-rank}. 

To discuss sign-rank, it is more convenient to switch from Boolean matrices to \emph{sign matrices}, which are matrices with $\pm 1$ entries.  The \emph{sign-rank} $\rank_\pm(F)$ of a sign matrix $F_{\cX \times \cY}$ is the smallest rank of a real matrix $A_{\cX \times \cY}$ such that the entries of $A$ are nonzero and have the same signs as their corresponding entries in $F$. Geometrically, sign-rank corresponds to the smallest dimension where we can represent $F$ as points and homogeneous half-spaces. 

We can reformulate the definition of sign-rank as follows.

\begin{definition}[Sign-rank]
\label{def:signrank}
The sign-rank of a sign matrix  $F_{\cX \times \cY}$ is the smallest $d$ such that there exist vectors $u_x, v_y\in \R^d$ with $F(x,y)=\sign(\inp{u_x}{v_y})$ for all  $(x,y) \in \cX\times \cY$.
\end{definition} 

Recall that the log-rank conjecture speculates that for deterministic protocols, the communication complexity is polynomially related to the logarithm of the rank of the corresponding matrix. Paturi and Simon proved that a similar and tighter connection is true for unbounded-error protocols, except that rank is replaced by sign-rank.
	
\begin{theorem}[Paturi and Simon~\cite{paturi1986probabilistic}]
	\label{thm:PaturiSimon}
	For every sign-matrix $F$, we have 
	\[\UU(F)= \log\rank_\pm(F) \pm O(1). \]
\end{theorem}	

In light of \cref{thm:PaturiSimon}, to study $\UU(F)$, one can set aside the intricacies of communication and focus on the geometric notion of sign-rank.

\paragraph{Number of matrices of small sign-rank:} Shortly after the introduction of sign-rank in~\cite{paturi1986probabilistic}, Alon, Frankl, and R\"odl~\cite{alon1985geometrical} used results of~\cite{MR161339,MR0200942,MR226281} on the number of connected components of real algebraic varieties and obtained a linear lower bound on the sign-rank of random sign matrices. This argument was later refined in~\cite[Lemma 24]{DBLP:conf/colt/AlonMY16} to the following bound on the number of low sign-rank matrices. 

\begin{lemma}[{See \cite[Lemma 24]{DBLP:conf/colt/AlonMY16}}]
\label{lem:count}
For  $d \le \frac{m}{2}$, the number of $m \times m$ sign matrices of sign-rank at most $d$ does not exceed $\left( O(m/d)\right)^{2dm} \leq 2^{O(dm \log(m))}$.
\end{lemma}
\cref{lem:count} shows that there are very few matrices with small sign-rank and that a typical $m\times m$ sign matrix has sign-rank $\Omega(m)$. This scarcity of small-sign-rank matrices suggests that they might possess strong structural properties.

\paragraph{Large monochromatic rectangles:}  \cite{ALON2005310} used the geometric properties of sign-rank to prove  that every $\cX \times \cY$ sign matrix of  sign-rank $d$ contains an  $\frac{|\cX|}{2^{d+1}} \times \frac{|\cY|}{2^{d+1}}$ monochromatic rectangle.  Their result uses a theorem of Yao and Yao~\cite{YaoYao}, which is based on the Borsuk-Ulam theorem, a result in topology.  Slightly different bounds are also obtained in~\cite{fox2016polynomial} using the cutting lemma of Chazelle~\cite{MR1194032}.  In our notation, we have the following relation between sign-rank and $\wrc(F)$. 

\begin{theorem}[See {\cite[Theorem 1.3]{ALON2005310}}]
\label{thm:monochromaticSignRank}
For every sign-matrix $F$, we have
\begin{equation}
\label{eq:AlonRect}
\rank_\pm(F) \gtrsim  \log \left( \wrc(F)^{-1} \right). 
\end{equation}
\end{theorem}

On the other hand, \cite{hatami2022lower} used a counting argument to show that there are matrices with $\wrc(A)^{-1}=O(1)$ and very large sign-rank.  
 
\begin{theorem}[{\cite[Theorem 3.2]{hatami2022lower}}] 
\label{thm:count_signrank_large}
There exist $m \times m$ sign matrices $A$ such that  
\[
\wrc(A)^{-1} \le 2^{15}, \qquad \text{while} \qquad \rank_\pm(A)=\Omega\left(\frac{m^{1/3}}{\log(m)}\right).
\]
\end{theorem}
It is known that $\Pcc^{\mathsf{NP}} \subsetneq \UPP$ (see~\cite{goos_landscape}).
Therefore, \cref{thm:count_signrank_large} is stronger than the separation of \cref{thm:GapPNP}, as it shows the existence of communication problems with $\wrc(F)^{-1}=O(1)$ that are not in $\UPP$. This answers an open problem by G\"o\"os, Kamath, Pitassi, and Watson~\cite{goos2019query}.

The caveat of \cref {thm:count_signrank_large} is that its existential proof does not provide any explicit construction. In fact, regarding explicit examples that separate sign-rank and $\wrect(\cdot)^{-1}$, our knowledge is embarrassingly limited. The following problem is open. 
  
\begin{problem}
\label{prob:explicit}
Construct an explicit sequence of matrices $F_n$ such that $\wrect(F_n)^{-1}=O(1)$ and
\[ \lim_{n \to \infty} \rank_{\pm}(F_n)=\infty. \]
\end{problem}

By \cref{thm:count_signrank_large}, we know such matrices exist, and in fact, with very large sign-ranks. On the other hand,  none of the known lower bound techniques can directly imply a solution to \cref{prob:explicit}. Indeed, in addition to the monochromatic rectangle lower bound, there are only two other known methods for
proving lower bounds on the sign-rank of explicit matrices: (i) Sign-rank is at least the VC-dimension: $\rank_\pm(A) \ge \mathrm{VC}(A)$; (ii) Forster’s method, which states that sign-rank is at least the inverse of the largest possible average margin among the representations of the matrix by points and half-spaces:  
 $\rank_\pm(F) \ge \margin^{\avg}(F)^{-1}$. We refer the reader to~\cite{hatami2022lower} for the definition of average margin and a thorough discussion of these facts. 

Qualitatively, \Cref{eq:AlonRect} is the strongest known method for proving lower bounds on the sign-rank of an explicit matrix.  If it fails to provide a super-constant lower bound for the sign-rank of a  matrix, then the other two methods will also fail. More precisely, we have 
\begin{equation}
\label{eq:rectanglestrong}
\sqrt{\VC(A)} \le \amargin(A)^{-1} \le \wrc^{-1}(A).
\end{equation} 
In this sense, \cref{prob:explicit} captures the limitation of the currently known lower bound techniques for sign-rank.

\paragraph{Sign-rank of hypercubes and $\BPP_0 \text{ vs } \UPP_0$:} Linial, Mendelson, Schechtman, and Shraibman~\cite{MR2359826} asked whether sign-rank can be bounded from above by a function of the so-called margin complexity. The relations between margin complexity, discrepancy, and randomized communication complexity, which were discovered later, allows us to rephrase their question as follows. 

\begin{question}
\label{q:LMSS}
Is it true that $\BPP_0 \subseteq \UPP_0$? Equivalently, is it possible to upper bound $\rank_\pm(F)$ by a function of $\Rcc(F)$? 
\end{question}

We believe the answer to this question to be negative.  Consider the sign version of the \textsc{Hypercube} problem, that is, let $\Q_n$ be the sign matrix whose rows and columns are indexed with the elements of $\{0,1\}^n$, and $\Q_n(x,y)=-1$ if $x$ and $y$ differ in exactly one coordinate. As we discussed earlier, $\Rcc(\Q_n)=O(1)$. We conjecture that the sign-rank of $\Q_n$ tends to infinity as $n$ grows, which, if true, would imply $\BPP_0 \not\subseteq \UPP_0$.   

\begin{conjecture}[Sign-rank of \textsc{Hypercube}~\cite{hatami2022lower} \textcolor{red}{Refuted in \cite{GoosHarmsImbachSokolov2025}}] 
\label{conj:cubes}
We have 
\[\lim_{n \to \infty} \rank_\pm(\Q_n)=\infty.\]
\end{conjecture} 

\vskip5pt 
\begin{updatebox}
 G\"o\"os, Harms, Imbach, and Sokolov~\cite{GoosHarmsImbachSokolov2025} refuted this conjecture, proving that $\Q_n$ has constant sign-rank. The following paragraph is no longer relevant.
\end{updatebox}

It is worth pointing out that proving \cref{conj:cubes} in the positive would likely require some new lower bound techniques, as $\wrect(\Q_n)^{-1}=O(1)$, shown in~\cite{hatami2022lower}. Note that a positive answer to \cref{conj:cubes} would also solve \cref{prob:explicit}.

\paragraph{Equality oracles, $\Pcc_0^\EQ  \subsetneq \UPP_0$:}
It is easy to show that $\EQ_n\in \UPP_0$, as its sign-rank is $3$. 
The following theorem shows that, in fact, $\Pcc_0^\EQ  \subseteq\UPP_0$. The example of \textsc{Greater-Than} shows that this inclusion is strict. It is easy to see that $\textsc{Greater-Than}  \in \UPP_0$ as its sign-rank is $2$. This fact combined with $\Rcc(\GT_n)=\Theta(\log n)$ shows that  
\[ \textsc{Greater-Than}  \in \UPP_0 \setminus  \BPP_0  \subseteq  \UPP_0 \setminus  \Pcc_0^\EQ.\]
\begin{theorem}[\cite{hatami2022lower}]
\label{eq:EqualityOracle}
For every sign matrix $F_{\cX \times \cY}$, we have $\rank_{\pm}(F) \leq 4^{\DD^{\EQ}(F)}$. In particular, $\UU(F) \le 2 \DD^\EQ(F) + O(1)$.  
\end{theorem}
\begin{proof} 
    We proceed by induction on $d \coloneqq \DD^{\EQ}(F)$. When $d = 1$, $F$ corresponds to a blocky matrix, which in fact has $\rank_\pm(F)\leq 3$. For larger $d$, consider a cost $d$ protocol for a sign matrix $F$ and suppose the equality query at the root of the tree is $\EQ(a(x), b(y))$, where here we assume without loss of generality $a(x)$ and $b(y)$ take integer values. Let $S_{\cX\times \cY}$ be the matrix with entries $S_{xy}=\1_{a(x)=b(y)}$.
    We branch according to the output of the first query either to the left or the right subtree of the root, each corresponding to a protocol with cost at most $d-1$.  Let the corresponding matrices for these protocols be $\Pi_1$ and $\Pi_2$, and note that  
    \[
    F = S \circ \Pi_1 + (\J-S) \circ \Pi_2,
    \]     
    where $\J \defeq \J_{\cX \times \cY}$ is the all-ones matrix.  
    By  the induction hypothesis, $\Pi_1$ and $\Pi_2$ have sign-rank at most  $\leq 4^{d-1}$. Let $\widetilde{\Pi}_1$ and $\widetilde{\Pi}_2$ be real matrices with rank at most $4^{d-1}$ that satisfy $\sign(\widetilde{\Pi}_1)=\Pi_1$ and $\sign(\widetilde{\Pi}_2)=\Pi_2$. 
    Let $E_{\cX\times \cY}$ be the rank-$3$ matrix with entries $E_{xy}= (a(x)-b(y))^2$. Note that for a sufficiently large $k$, we have 
    \[
    F=\sign(\widetilde{\Pi}_1 + k E \circ \widetilde{\Pi}_2 ). 
    \]
   Finally, we have
       \[
    \rank(\widetilde{\Pi}_1 + k \widetilde{\Pi}_2 \circ E) \le \rank(\widetilde{\Pi}_1) + \rank(\widetilde{\Pi}_2)\cdot \rank(E)\le 4^{d-1}+ 3\cdot 4^{d-1} = 4^{d}.\qedhere \]  
\end{proof}

The above theorem combined with \cref{eq:DEQvsR} shows that $\Pcc_0^\EQ \subseteq \UPP_0\cap \BPP_0$. Both the inclusions $\Pcc_0^\EQ \subsetneq \UPP_0$ and $\Pcc_0^\EQ \subsetneq \BPP_0$ are strict; the former follows from the example of \textsc{Greater-Than}  and the latter holds for \textsc{Hypercube}~\cite{HHH23,HWZ22}. The question of whether these separations can be obtained simultaneously was asked recently by \cite{harms2023randomized}.
\begin{question}[\cite{harms2023randomized} \textcolor{red}{Solved: Refuted in~\cite{GoosHarmsImbachSokolov2025}}]
\label{q:Harms}
Is it the case that $\Pcc_0^\EQ = \UPP_0\cap \BPP_0$?
\end{question}

\begin{updatebox}
 This question is refuted by G\"o\"os, Harms, Imbach, and Sokolov~\cite{GoosHarmsImbachSokolov2025} as they proved that $\Q_n$ has constant sign-rank;  recall that it also  satisfies $\Rcc(\Q_n) = O(1)$ and $\DD^{\EQ}(\Q_n) = \Theta(\log n)$. More strongly, the point-line arrangement of Goh and Hatami~\cite{gohhata2026} has constant sign-rank, and satisfies $\Rcc(F) = O(1)$ and $\DD^{\EQ}(F) = \Theta(n)$. The following two paragraphs are no longer relevant.
\end{updatebox}

Since $\textsc{Hypercube} \not\in \Pcc_0^\EQ$ and $\textsc{Hypercube} \in \BPP_0$, a positive answer to \cref{q:Harms} would imply $\textsc{Hypercube} \not\in \UPP_0$ and solve \cref{conj:cubes}. 

In the converse direction, \cref{conj:cubes}  would imply a positive answer to \cref{q:Harms} for the special case of \textsc{xor}-lifts. Indeed, if $\textsc{Hypercube} \not\in \UPP_0$, then the result of \cite{cheung2022boolean} would imply that every family of  \textsc{xor}-lift $f_n^\oplus \in \UPP_0 \cap \BPP_0$ must satisfy $\norm{f_n}_A = O(1)$ and therefore by \cref{{thm:Cayley}} and \cref{prop:EqOracle}, we must have $f_n^\oplus \in \Pcc_0^\EQ$.

\subsection{Weakly unbounded-error complexity, $\PPcc$}
The \emph{weakly unbounded-error communication complexity} of a  problem $F_{\cX\times \cY}$ is defined as 
\[
\mathsf{PP}(F)\defeq \min_{\epsilon<1/2} \Rcc_\epsilon(F)+\log\left(\frac{1}{\frac{1}{2}-\epsilon}\right),
\]
where $\Rcc_\epsilon(F)$ is the minimum cost of a public-coin randomized protocol with two-sided error at most $\epsilon$.  
 
Define $\PPcc$ and $\PPcc_0$ to be the classes of families of $n$-bit problems $F_n$ with $\mathsf{PP}(F_n)=\polylog n$ and $\mathsf{PP}(F_n)=O(1)$, respectively. It is immediate from the definitions and Newman's lemma~\cite{newman} that 
\[
\BPP\subseteq \PPcc \subseteq \UPP,\qquad \text{and} \qquad \PPcc_0=\BPP_0.
\]

As discussed before, the \textsc{Greater-Than} problem separates $\UPP_0$ from $\BPP_0=\PPcc_0$. Babai, Frankl, and Simon~\cite{bfs86} asked whether $\UPP=\PPcc$. Their question remained unanswered for over two decades, until ~\cite{buhrman2007computation,sherstov2008halfspace} independently showed that there are $2^n \times 2^n$ sign matrices $F$ with $\rank_\pm(F)= n$ but $\PPcc(F) =  n^{\Omega(1)}$. The separation was strengthened in subsequent works to $\rank_\pm(F) = n$ and $\PPcc(F)= \Omega(n)$ in \cite{MR4298067}.

The example proposed in \cite{buhrman2007computation}, in fact, belongs to $\Pcc^\NP$, and therefore also shows that $\Pcc^\NP\not\subseteq \PPcc$. It turns out that the opposite direction of this inclusion is not true either, as follows from an argument involving bounds on the rectangle ratio. 
\begin{theorem}
There exists a family of $n$-bit communication problems $F_n$ with $\PPcc(F_n)=O(\log(n))$ and $\rect(F_n)^{-1}=2^{\Omega(n)}$. In particular, $\DD^\NP(F_n)=\Omega(n)$ and $\PPcc \not\subseteq \Pcc^\NP$.
\end{theorem}
\begin{proof}
Let $F_n(x,y)= 1$ iff $x$ and $y$ differ on at least $n/2$ bits. It is known through classical results from combinatorics~\cite{frankl1981short} that $\rect(F_n)^{-1}=2^{\Omega(n)}$. Thus by \cref{thm:IW}, we get $\DD^\NP(F_n)=\Omega(n)$.

To see the inclusion in $\PPcc$, note that the simple protocol where two parties pick a random index $i$ uniformly at random and output $1$ iff $x_i\neq y_i$, has cost $O(\log n)$.
\end{proof}

Finally, recent works have shown that $\PPcc$ does not even contain $\UPP_0$. Indeed, \cite{MR4129280,ahmed2023communication} gave simple constructions of $n$-bit communication problems $F$ with $\rank_\pm(F)=3$ and $\PPcc(F)=\Omega(n)$.

\section{Final remarks}

We discussed several open problems that indicate significant gaps in our understanding of communication complexity and capture the limitations of the currently available techniques. For example, disproving $\BPP = \Pcc^\RP$, $\BPP \subseteq \Pcc^\NP$, or giving a negative answer to \cref{q:BPP_large_rect}  or \cref{q:uniformrectangle} requires constructing a family of matrices in $\BPP$ that is fundamentally different from all the currently known examples. Conversely, proving that any of these statements is true would be a major stride toward achieving a structural description of $\BPP$. Similarly, \cref{q:LMSS} and \cref{prob:explicit} require a new lower-bound technique for sign-rank that can reach beyond the  $\log \wrect(\cdot)^{-1}$  bound of \cref{thm:monochromaticSignRank}.

We hope that, similar to the introduction of communication classes by~\cite{bfs86}, the formal paradigm of constant-cost communication classes will catalyze future research, and efforts to establish separations between these classes will lead to the discovery of new examples and lower-bound techniques and give us a deeper understanding of communication models and their connections to other areas of theoretical computer science. 

Due to space limitations, we did not discuss quantum communication models, multi-party models, search problems, and various related query models, most notably parity decision trees. The questions that are being discussed in this article can be asked in a similar way for these models.

We conclude by presenting, in \cref{fig1} and \cref{fig:enter-label}, the known relations and separations among various classes discussed in this article. These figures include a selected list of classes, excluding easier-to-handle classes such as $\Pcc, \RP, \NP, \Pcc_0=\NP_0$.  We define the classes $\Rect$ and $\Rect_0$ to consist of matrix families with $\wrect^{-1}(\cdot)$ bounded from above by $2^{\polylog n}$ and $O(1)$, respectively. The class $\Rect$ appears in \cite{goos2019query} with the different name of $\mathsf{PM}$ (for Product Method).

\begin{figure}[H]
    \centering
\[
\begin{array}{c|cccccccccccc}
    & \Pcc_0^{\EQ} & \Pcc_0^{\RP} & \BPP_0 & \UPP_0 & \Rect_0 & \Pcc^{\EQ} & \Pcc^{\RP} & \BPP & \PPcc & \UPP & \Pcc^{\NP} & \Rect \\
\hline
\Pcc_0^{\EQ} & = & \subseteq & \subseteq & \subseteq & \subseteq & \subseteq & \subseteq & \subseteq & \subseteq & \subseteq & \subseteq & \subseteq \\
\Pcc_0^{\RP} & \not\subseteq & = & \subseteq & {\textbf{\color{red}{\circled{?}}}}  &  \subseteq & \not\subseteq  & \subseteq & \subseteq & \subseteq & \subseteq & \subseteq & \subseteq \\
\BPP_0 & \not\subseteq & {\textbf{\color{red}{\circled{?}}}}  & = & {\textbf{\color{red}{\circled{?}}}} & {\textbf{\color{red}{\circled{?}}}}  & \not\subseteq  & {\textbf{\color{red}{\circled{?}}}}  & \subseteq & \subseteq & \subseteq & {\textbf{\color{red}{\circled{?}}}}  & {\textbf{\color{red}{\circled{?}}}}  \\
\UPP_0 & \not\subseteq & \not\subseteq & \not\subseteq & = & \subseteq & \not\subseteq & \not\subseteq & \not\subseteq & \not\subseteq & \subseteq & {\textbf{\color{red}{\circled{?}}}}  & \subseteq \\
\Rect_0 & \not\subseteq & \not\subseteq & \not\subseteq & \not\subseteq & = & \not\subseteq & \not\subseteq & \not\subseteq & \not\subseteq & \not\subseteq & \not\subseteq & \subseteq \\
\Pcc^{\EQ} & \cellcolor{lightgray}{\color{White}{\not\subseteq}} & \cellcolor{lightgray}{\color{White}{\not\subseteq}} & \cellcolor{lightgray}{\color{White}{\not\subseteq}} & \cellcolor{lightgray}{\color{White}{\not\subseteq}} & \cellcolor{lightgray}{\color{White}{\not\subseteq}} & = & \subseteq & \subseteq & \subseteq & \subseteq & \subseteq & \subseteq \\
\Pcc^{\RP} & \cellcolor{lightgray}{\color{White}{\not\subseteq}} & \cellcolor{lightgray}{\color{White}{\not\subseteq}} & \cellcolor{lightgray}{\color{White}{\not\subseteq}} & \cellcolor{lightgray}{\color{White}{\not\subseteq}} & \cellcolor{lightgray}{\color{White}{\not\subseteq}} & \not\subseteq & = & \subseteq & \subseteq & \subseteq & \subseteq & \subseteq \\
\BPP & \cellcolor{lightgray}{\color{White}{\not\subseteq}} & \cellcolor{lightgray}{\color{White}{\not\subseteq}} & \cellcolor{lightgray}{\color{White}{\not\subseteq}} & \cellcolor{lightgray}{\color{White}{\not\subseteq}} & \cellcolor{lightgray}{\color{White}{\not\subseteq}} & \not\subseteq & {\textbf{\color{red}{\circled{?}}}}  & = & \subseteq & \subseteq & {\textbf{\color{red}{\circled{?}}}}  & {\textbf{\color{red}{\circled{?}}}}  \\
\PPcc & \cellcolor{lightgray}{\color{White}{\not\subseteq}} & \cellcolor{lightgray}{\color{White}{\not\subseteq}} & \cellcolor{lightgray}{\color{White}{\not\subseteq}} & \cellcolor{lightgray}{\color{White}{\not\subseteq}} & \cellcolor{lightgray}{\color{White}{\not\subseteq}} & \not\subseteq & \not\subseteq & \not\subseteq & = &  \subseteq & \not\subseteq & \not\subseteq \\
\UPP & \cellcolor{lightgray}{\color{White}{\not\subseteq}} & \cellcolor{lightgray}{\color{White}{\not\subseteq}} & \cellcolor{lightgray}{\color{White}{\not\subseteq}} & \cellcolor{lightgray}{\color{White}{\not\subseteq}} & \cellcolor{lightgray}{\color{White}{\not\subseteq}} & \not\subseteq & \not\subseteq & \not\subseteq & \not\subseteq & = & \not\subseteq & \not\subseteq \\
\Pcc^{\NP} & \cellcolor{lightgray}{\color{White}{\not\subseteq}} & \cellcolor{lightgray}{\color{White}{\not\subseteq}} & \cellcolor{lightgray}{\color{White}{\not\subseteq}} & \cellcolor{lightgray}{\color{White}{\not\subseteq}} & \cellcolor{lightgray}{\color{White}{\not\subseteq}} & \not\subseteq & \not\subseteq & \not\subseteq & \not\subseteq & \subseteq & = & \subseteq \\
\Rect & \cellcolor{lightgray}{\color{White}{\not\subseteq}} & \cellcolor{lightgray}{\color{White}{\not\subseteq}} & \cellcolor{lightgray}{\color{White}{\not\subseteq}} & \cellcolor{lightgray}{\color{White}{\not\subseteq}} & \cellcolor{lightgray}{\color{White}{\not\subseteq}} & \not\subseteq & \not\subseteq & \not\subseteq & \not\subseteq & \not\subseteq & \not\subseteq & = \\
\end{array}
\]
     \caption{The entry at a row $\mathsf{A}$ and a column $\mathsf{B}$ indicates whether $\mathsf{A}\subseteq \mathsf{B}$ or $\mathsf{A}\not\subseteq \mathsf{B}$. A question mark indicates that the relationship is unknown. The separations in grey entries follow trivially via padding. \begin{updatebox} The separations $\BPP_0 \not\subseteq \Pcc^{\EQ}$ and $\Pcc_0^{\RP} \not\subseteq \Pcc^{\EQ}$ are new and due to~\cite{GoosHarmsRiazanov2025}. See also~\cite{gohhata2026} for stronger bounds.\end{updatebox}}
    \label{fig1}
\end{figure}

\begin{figure}[H]
    \centering
    
\tikzset{every picture/.style={line width=0.75pt}} %set default line width to 0.75pt        

\begin{tikzpicture}[x=0.75pt,y=0.75pt,yscale=-.7,xscale=.9]
%uncomment if require: \path (0,408); %set diagram left start at 0, and has height of 408

%Shape: Rectangle [id:dp8949343972911024] 
\draw   (171.5,225) -- (220.5,225) -- (220.5,257) -- (171.5,257) -- cycle ;
%Shape: Rectangle [id:dp0850462524246769] 
\draw   (120,283) -- (164,283) -- (164,314) -- (120,314) -- cycle ;
%Shape: Rectangle [id:dp4538502764742627] 
\draw   (39,341) -- (81,341) -- (81,370) -- (39,370) -- cycle ;
%Shape: Rectangle [id:dp968385290737144] 
\draw   (40.5,162) -- (91.5,162) -- (91.5,195) -- (40.5,195) -- cycle ;
%Straight Lines [id:da9272981359860903] 
\draw    (81,341) -- (117.53,315.71) ;
\draw [shift={(120,314)}, rotate = 145.3] [fill={rgb, 255:red, 0; green, 0; blue, 0 }  ][line width=0.08]  [draw opacity=0] (8.93,-4.29) -- (0,0) -- (8.93,4.29) -- cycle    ;
%Straight Lines [id:da722296234991815] 
\draw    (164,283) -- (195.13,258.84) ;
\draw [shift={(197.5,257)}, rotate = 142.18] [fill={rgb, 255:red, 0; green, 0; blue, 0 }  ][line width=0.08]  [draw opacity=0] (8.93,-4.29) -- (0,0) -- (8.93,4.29) -- cycle    ;
%Straight Lines [id:da7485673708633149] 
\draw    (61.5,340) -- (63.46,200) ;
\draw [shift={(63.5,197)}, rotate = 90.8] [fill={rgb, 255:red, 0; green, 0; blue, 0 }  ][line width=0.08]  [draw opacity=0] (8.93,-4.29) -- (0,0) -- (8.93,4.29) -- cycle    ;
%Shape: Rectangle [id:dp8424114314830737] 
\draw   (167,27) -- (223,27) -- (223,59) -- (167,59) -- cycle ;
%Straight Lines [id:da09509837960462897] 
\draw    (195.5,225) -- (195.5,62) ;
\draw [shift={(195.5,59)}, rotate = 90] [fill={rgb, 255:red, 0; green, 0; blue, 0 }  ][line width=0.08]  [draw opacity=0] (8.93,-4.29) -- (0,0) -- (8.93,4.29) -- cycle    ;
%Shape: Rectangle [id:dp4419149180638402] 
\draw   (319.5,342) -- (359.5,342) -- (359.5,370) -- (319.5,370) -- cycle ;
%Shape: Rectangle [id:dp29571573592806333] 
\draw   (320,184) -- (362,184) -- (362,212) -- (320,212) -- cycle ;
%Straight Lines [id:da5743953036018963] 
\draw    (81,355) -- (317,355.99) ;
\draw [shift={(320,356)}, rotate = 180.24] [fill={rgb, 255:red, 0; green, 0; blue, 0 }  ][line width=0.08]  [draw opacity=0] (8.93,-4.29) -- (0,0) -- (8.93,4.29) -- cycle    ;
%Shape: Rectangle [id:dp8944744146062562] 
\draw   (322,111) -- (362,111) -- (362,138) -- (322,138) -- cycle ;
%Shape: Rectangle [id:dp17165005675310963] 
\draw   (320,27) -- (369,27) -- (369,57) -- (320,57) -- cycle ;
%Shape: Rectangle [id:dp38713988608586014] 
\draw   (495,30) -- (543,30) -- (543,58) -- (495,58) -- cycle ;
%Shape: Rectangle [id:dp7696107045464466] 
\draw   (498.5,234) -- (539.5,234) -- (539.5,261) -- (498.5,261) -- cycle ;
%Shape: Rectangle [id:dp564473242710454] 
\draw   (496,112) -- (541,112) -- (541,142) -- (496,142) -- cycle ;
%Straight Lines [id:da9397154267977237] 
\draw    (360.5,355) -- (496.02,262.69) ;
\draw [shift={(498.5,261)}, rotate = 145.74] [fill={rgb, 255:red, 0; green, 0; blue, 0 }  ][line width=0.08]  [draw opacity=0] (8.93,-4.29) -- (0,0) -- (8.93,4.29) -- cycle    ;
%Straight Lines [id:da4416088056703833] 
\draw    (519,112) -- (519.47,61) ;
\draw [shift={(519.5,58)}, rotate = 90.53] [fill={rgb, 255:red, 0; green, 0; blue, 0 }  ][line width=0.08]  [draw opacity=0] (8.93,-4.29) -- (0,0) -- (8.93,4.29) -- cycle    ;
%Straight Lines [id:da8630033433988109] 
\draw    (342,111) -- (342,60) ;
\draw [shift={(342,57)}, rotate = 90] [fill={rgb, 255:red, 0; green, 0; blue, 0 }  ][line width=0.08]  [draw opacity=0] (8.93,-4.29) -- (0,0) -- (8.93,4.29) -- cycle    ;
%Straight Lines [id:da9883029679540984] 
\draw    (518.5,233) -- (518.98,145) ;
\draw [shift={(519,142)}, rotate = 90.31] [fill={rgb, 255:red, 0; green, 0; blue, 0 }  ][line width=0.08]  [draw opacity=0] (8.93,-4.29) -- (0,0) -- (8.93,4.29) -- cycle    ;
%Straight Lines [id:da9832267178269642] 
\draw    (221.5,240) -- (317.22,201.13) ;
\draw [shift={(320,200)}, rotate = 157.9] [fill={rgb, 255:red, 0; green, 0; blue, 0 }  ][line width=0.08]  [draw opacity=0] (8.93,-4.29) -- (0,0) -- (8.93,4.29) -- cycle    ;
%Straight Lines [id:da21398778650784678] 
\draw    (498,247) -- (364.84,200.98) ;
\draw [shift={(362,200)}, rotate = 19.06] [fill={rgb, 255:red, 0; green, 0; blue, 0 }  ][line width=0.08]  [draw opacity=0] (8.93,-4.29) -- (0,0) -- (8.93,4.29) -- cycle    ;
%Straight Lines [id:da5305931326519899] 
\draw    (69.5,162) -- (317.23,58.16) ;
\draw [shift={(320,57)}, rotate = 157.26] [fill={rgb, 255:red, 0; green, 0; blue, 0 }  ][line width=0.08]  [draw opacity=0] (8.93,-4.29) -- (0,0) -- (8.93,4.29) -- cycle    ;
%Curve Lines [id:da31343737336063193] 
\draw [color={rgb, 255:red, 0; green, 0; blue, 0 }  ,draw opacity=1 ]   (223,59) .. controls (317.53,99.8) and (385.32,105.94) .. (493.37,58.72) ;
\draw [shift={(495,58)}, rotate = 156.23] [fill={rgb, 255:red, 0; green, 0; blue, 0 }  ,fill opacity=1 ][line width=0.08]  [draw opacity=0] (9.82,-4.72) -- (0,0) -- (9.82,4.72) -- cycle    ;
%Straight Lines [id:da28209144454548185] 
\draw    (496,112) -- (371.75,58.19) ;
\draw [shift={(369,57)}, rotate = 23.42] [fill={rgb, 255:red, 0; green, 0; blue, 0 }  ][line width=0.08]  [draw opacity=0] (8.93,-4.29) -- (0,0) -- (8.93,4.29) -- cycle    ;
%Straight Lines [id:da561601487082944] 
\draw    (69.5,162) -- (164.94,61.18) ;
\draw [shift={(167,59)}, rotate = 133.43] [fill={rgb, 255:red, 0; green, 0; blue, 0 }  ][line width=0.08]  [draw opacity=0] (8.93,-4.29) -- (0,0) -- (8.93,4.29) -- cycle    ;
%Straight Lines [id:da06347362258162437] 
\draw    (342.5,183) -- (342.97,141) ;
\draw [shift={(343,138)}, rotate = 90.64] [fill={rgb, 255:red, 0; green, 0; blue, 0 }  ][line width=0.08]  [draw opacity=0] (8.93,-4.29) -- (0,0) -- (8.93,4.29) -- cycle    ;

% Text Node
\draw (175.04,231.3) node [anchor=north west][inner sep=0.75pt]  [rotate=-0.26]  {$\mathsf{BPP}_{0}$};
% Text Node
\draw (127.04,287.3) node [anchor=north west][inner sep=0.75pt]  [rotate=-0.26]  {$\mathsf{P}_{0}^{\mathsf{RP}}$};
% Text Node
\draw (47.04,344.3) node [anchor=north west][inner sep=0.75pt]  [rotate=-0.26]  {$\mathsf{P}_{0}^{\mathtt{EQ}}$};
% Text Node
\draw (45.04,168.3) node [anchor=north west][inner sep=0.75pt]  [rotate=-0.26]  {$\mathsf{UPP}_{0}$};
% Text Node
\draw (174.04,32.3) node [anchor=north west][inner sep=0.75pt]  [rotate=-0.26]  {$\mathsf{Rect}_{0}$};
% Text Node
\draw (326.48,346.41) node [anchor=north west][inner sep=0.75pt]  [rotate=-0.26]  {$\mathsf{P}^{\mathtt{EQ}}$};
% Text Node
\draw (323.98,189.41) node [anchor=north west][inner sep=0.75pt]  [rotate=-0.26]  {$\mathsf{BPP}$};
% Text Node
\draw (330.04,115.3) node [anchor=north west][inner sep=0.75pt]  [rotate=-0.26]  {$\mathsf{PP}$};
% Text Node
\draw (327.04,33.3) node [anchor=north west][inner sep=0.75pt]  [rotate=-0.26]  {$\mathsf{UPP}$};
% Text Node
\draw (502.04,35.3) node [anchor=north west][inner sep=0.75pt]  [rotate=-0.26]  {$\mathsf{Rect}$};
% Text Node
\draw (504.04,238.3) node [anchor=north west][inner sep=0.75pt]  [rotate=-0.26]  {$\mathsf{P}^{\mathsf{RP}}$};
% Text Node
\draw (504.04,117.3) node [anchor=north west][inner sep=0.75pt]  [rotate=-0.26]  {$\mathsf{P}^{\mathsf{NP}}$};

\end{tikzpicture}

\caption{$\mathsf{A}\rightarrow \mathsf{B}$ indicates $\mathsf{A}\subseteq \mathsf{B}$.}
    \label{fig:enter-label}
\end{figure}

\paragraph{Acknowledgements.} We wish to thank Tsun-Ming Cheung, Mika G\"o\"os, Lianna Hambardzumyan, Nathan Harms, 
 Shachar Lovett, and Alexander Razborov for their helpful comments on a draft of this article. We thank Anthony Ostuni for the erratum in \cref{sec:analyticrank}.

\bibliographystyle{alpha}
\bibliography{mybib.bib}

\end{document}

%% file: myheader.tex
\usepackage{amsmath,amsthm,amssymb,amscd,amstext,amsfonts,mathscinet}
\usepackage{mathtools}
\usepackage{mathrsfs}
\usepackage{thmtools}
\usepackage{latexsym}
\usepackage{verbatim}
\usepackage{framed}
\usepackage{graphicx}
\usepackage{enumitem}
\usepackage{fullpage}
\usepackage{bm}
\usepackage{hyperref}
\usepackage{url}
\usepackage{physics}
\usepackage{dsfont}
\usepackage{hyperref}
\hypersetup{
    colorlinks=true,
    linkcolor=blue,
    filecolor=blue,      
    urlcolor=blue,
    citecolor=blue,  
    pdfborder={0 0 0},
}
\usepackage{float}

\usepackage{graphicx}
\usepackage{amsfonts}
\usepackage{amsmath}
\usepackage[usenames,dvipsnames,svgnames,table]{xcolor}
\usepackage{tikz}
\usepackage{mathdots}
\usepackage{yhmath}
\usepackage{cancel}
\usepackage{siunitx}

\usepackage{array}
\usepackage{multirow}
\usepackage{gensymb}
\usepackage{tabularx}
\usepackage{extarrows}
\usepackage{booktabs}
\usetikzlibrary{fadings}
\usetikzlibrary{patterns}
\usetikzlibrary{shadows.blur}
\usetikzlibrary{shapes}
\AtBeginDocument{\RenewCommandCopy\qty\SI}
\ExplSyntaxOn
\msg_redirect_name:nnn { siunitx } { physics-pkg } { none }
\ExplSyntaxOff

\newcommand*\circled[1]{\tikz[baseline=(char.base)]{
             \node[shape=circle,draw,inner sep=2pt] (char) {{#1}};}}

\usepackage{multicol}
\usepackage[nameinlink, noabbrev, capitalize]{cleveref}
\theoremstyle{plain}

\newtheorem*{theorem*}{Theorem}
\newtheorem{prb}{Theorem}
\newtheorem{problem}[prb]{Problem}
\newtheorem{question}[prb]{Question}
\newtheorem{conjecture}[prb]{Conjecture}

\newtheorem{thm}{Theorem}[section]
\newtheorem{theorem}[thm]{Theorem}

\newtheorem{lemma}[thm]{Lemma}
\newtheorem{proposition}[thm]{Proposition}
\newtheorem*{proposition*}{Proposition}
\newtheorem{definition}[thm]{Definition}

\newtheorem*{claim*}{Claim}

\theoremstyle{remark}

\theoremstyle{remark}
\DeclareMathOperator{\polylog}{polylog}
\DeclareMathOperator{\avg}{avg}
\DeclareMathOperator{\margin}{m}    
\DeclareMathOperator{\amargin}{\margin^{\avg}}
\DeclareMathOperator{\VC}{VC}    
\DeclareMathOperator{\sign}{sgn}

\DeclareMathOperator{\blockyrank}{{\mathrm{rk}_\mathpzc{Blocky}}}

\newcommand{\zo}{\{0,1\}}

\newcommand{\Pcc}{\mathsf{P}}
\newcommand{\PPcc}{\mathsf{PP}}

\newcommand{\SINK}{\mathtt{SINK}}
\newcommand{\rect}{\mathrm{rect}}
\newcommand{\wrc}{\mathrm{wrect}}
\newcommand{\Rcc}{\mathsf{R}}  
\newcommand{\Ncc}{\mathsf{N}}
\newcommand{\EQ}{\mathtt{EQ}}
\newcommand{\IP}{\mathtt{IP}}                       % inner product 

\newcommand{\GT}{\mathtt{GT}}
\newcommand{\BPP}{\mathsf{BPP}}
\newcommand{\ZPP}{\mathsf{ZPP}}
\newcommand{\RP}{{\mathsf{RP}}}
\newcommand{\UPP}{\mathsf{UPP}}
\newcommand{\NP}{{\mathsf{NP}}}
\newcommand{\coNP}{\mathsf{coNP}}
\newcommand{\coRP}{\mathsf{coRP}}

\renewcommand{\norm}[1]{\|#1\|}                     % norm
\newcommand{\inp}[2]{\left\langle#1,#2\right\rangle}            % inner product
\def\rank{{\mathrm{rk}}}                               % rank
\def\m{\mathbf{m}}

\newcommand{\cc}{{\mathsf{cc}}}

\newcommand{\C}{\mathbb{C}}

\newcommand{\F}{\mathbb{F}}

\newcommand{\N}{\mathbb{N}}
\newcommand{\Q}{\mathbb{Q}}
\newcommand{\R}{\mathbb{R}}

\newcommand{\Z}{\mathbb{Z}}

\newcommand{\cX}{\mathcal X}
\newcommand{\cY}{\mathcal Y}

\DeclareMathAlphabet{\mathpzc}{OT1}{pzc}{m}{it}

\DeclareMathOperator{\UU}{\mathsf{U}}
\newcommand{\defeq}{\coloneqq}
\def\DD{\mathsf{D}}     
\def\Ncc{\mathsf{N}}
\def\wrect{\mathrm{wrect}}

\def\Blocky{\mathpzc{Blocky}}

\def\intr{\mathtt{INT}} 
\def\I{\mathtt{I}}
\def\J{\mathtt{J}}
\def\Q{\mathtt{Q}}

\def\1{\mathbf{1}} 
\def\0{\mathbf{0}}

\let\latexchi\chi
\makeatletter
\renewcommand\chi{\@ifnextchar_\sub@chi\latexchi}
\newcommand{\sub@chi}[2]{% #1 is _, #2 is the subscript
  \@ifnextchar^{\subsup@chi{#2}}{\latexchi^{}_{#2}}%
}
\newcommand{\subsup@chi}[3]{% #1 is the subscript, #2 is ^, #3 is the superscript
  \latexchi_{#1}^{#3}%
}
\makeatother

\def\Rect{\mathsf{Rect}}      
\def\Blocky{\mathpzc{Blocky}}

\DeclareMathOperator{\IIP}{\mathtt{IIP}}          % for IC^ext